\documentclass[10pt,twocolumn]{IEEEtran} 

\usepackage{theorem} \usepackage{cite} \usepackage{stfloats} \usepackage{epsfig} \usepackage{verbatim}
\usepackage{graphicx} \usepackage{amssymb} \usepackage{amsmath} \usepackage{color}
\usepackage{bm}
\usepackage{xcolor}
\usepackage{algorithm}
\usepackage{algorithmic}

\newtheorem{remark}{Remark}

\theoremstyle{plain}
\newtheorem{theorem}{Theorem}
\theoremstyle{plain}
\newtheorem{proposition}{Proposition}
\theoremstyle{plain}

\theoremstyle{plain}
\newtheorem{lemma}{Lemma}

\newcommand\ve{{\bf e}}

\newcommand\vh{{\bf h}}

\newcommand\vq{{\bf q}}
\newcommand\vr{{\bf r}}
\newcommand\vs{{\bf s}}

\newcommand\vu{{\bf u}}

\newcommand\vy{{\bf y}}

\newcommand\mA{{\bf A}} 

\newcommand\mE{{\bf E}}

\newcommand\mH{{\bf H}}
\newcommand\mI{{\bf I}}

\newcommand\mR{{\bf R}}
\newcommand\mS{{\bf S}}
\newcommand\mT{{\bf T}}
\newcommand\mU{{\bf U}}
\newcommand\mV{{\bf V}}
\newcommand\mW{{\bf W}}
\newcommand\mX{{\bf X}}

\newcommand\mZ{{\bf Z}}

\renewcommand{\baselinestretch}{1}

\begin{document}

\def\QEDclosed{\mbox{\rule[0pt]{1.3ex}{1.3ex}}}

\def\QEDopen{{\setlength{\fboxsep}{0pt}\setlength{\fboxrule}{0.2pt}\fbox{\rule[0pt]{0pt}{1.3ex}\rule[0pt]{1.3ex}{0pt}}}}
\def\QED{\QEDopen}

\def\proof{}
\def\endproof{\hspace*{\fill}~\QED\par\endtrivlist\unskip}

\markboth{IEEE Transactions on Signal Processing, VOL. 64, No. 11, 2016}{YIN \MakeLowercase{\textit{et al.}}: Robust Pilot Decontamination Based on Joint Angle and Power Domain Discrimination}

\title{Robust Pilot Decontamination Based on Joint Angle and Power Domain Discrimination}
\author{Haifan~Yin,
    Laura~Cottatellucci,
    David~Gesbert~\IEEEmembership{Fellow,~IEEE},
    Ralf~R.~M{\"u}ller~\IEEEmembership{Senior Member,~IEEE},
    and~Gaoning~He%
\thanks{
Copyright (c) 2015 IEEE. Personal use of this material is permitted. However, permission to use this material for any other purposes must be obtained from the IEEE by sending a request to pubs-permissions@ieee.org.

This work was supported by Huawei France Research Center.  Ralf M{\"u}ller acknowledges the support of the Alexander von Humboldt-Foundation. A part of this work \cite{yin:15a} was presented on June 30, 2015 at the 16th IEEE International Workshop on Signal Processing Advances in Wireless Communications, SPAWC 2015, in Stockholm, Sweden.}%
\thanks{H. Yin, L. Cottatellucci, and  D. Gesbert are with EURECOM, 06410 Biot, France (e-mail: yin, cottatel, gesbert@eurecom.fr, ).}%
\thanks{R. R. M{\"u}ller is with Friedrich-Alexander Universit{\"a}t Erlangen-N{\"u}rnberg, Cauerstr. 7/LIT, 91058 Erlangen, Germany (e-mail: ralf.r.mueller@fau.de).}%
\thanks{G. He is with Huawei France Research Center, Huawei Technologies Co. Ltd., 92100 Boulogne-Billancourt, France (e-mail: hegaoning@huawei.com).}%
\thanks{Digital Object Identifier 10.1109/TSP.2016.2535204.}%
}

\maketitle

\begin{abstract}
We address the problem of noise and interference corrupted channel estimation in massive MIMO systems. Interference, which originates from pilot reuse (or contamination), can in principle be discriminated on the basis of the distributions of path angles and amplitudes. In this paper we propose novel robust channel estimation algorithms exploiting path diversity in both angle and power domains, relying on a suitable combination of the spatial filtering and amplitude based projection. The proposed approaches are able to cope with a wide range of system and topology scenarios, including those where, unlike in previous works, interference channel may overlap with desired channels in terms of multipath angles of arrival or exceed them in terms of received power. In particular we establish analytically the conditions under which the proposed channel estimator is fully decontaminated. Simulation results confirm the overall system gains when using the new methods.

\end{abstract}




\section{Introduction}\label{sec_intro}

Massive MIMO (also known as Large-Scale Antenna Systems) introduced in \cite{marzetta:10a}, is widely believed to be one of the key enablers of the future 5th generation (5G) wireless systems thanks to its potential to substantially enhance spectral and energy efficiencies \cite{marzetta:10a,ngo:13} compared to traditional MIMO with fewer antennas. This technique is based on the law of large numbers, which predicts that, as the number of base station antennas increases, the vector channel for a desired user terminal will grow more orthogonal to the vector channel of an interfering user, thus allowing the base station to reject interference by precoding, or even, as a low-complexity approach, simply aligning the beamforming vector with the desired channel (``Maximum Ratio Combining", or MRC), providing that Channel State Information (CSI) is known at base station. In practice however, CSI is acquired based on training sequences sent by user terminals. Due to limited time and frequency resources, non-orthogonal pilot sequences are typically used by user terminals in neighboring cells, resulting in residual channel estimation error. This effect, called \emph{pilot contamination} \cite{jose2009ISIT}, \cite{jose2011TWC}, has a detrimental impact on the actual achievable spectral and energy efficiencies in real systems.
As a result, considerable research efforts have been spent in the last couple of years towards alleviating pilot interference in massive MIMO networks.

Such techniques 
span from smart design of pilot reuse schemes (e.g. \cite{sorensen:14,atzeni:14}) to channel estimation techniques based on coordinated pilot allocation (e.g. \cite{yin:13,neumann:14}), to methods relying on multi-cell joint processing (e.g. \cite{ashikhmin:14}), to nonlinear channel estimation techniques leveraging on some fundamental features of massive MIMO systems (e.g. \cite{ngo:12,mueller:14,yin:13,hu2015}).

Two key features of massive MIMO channels that have been previously reported are of particular interest here: 1) channels of different users tend to be pairwise orthogonal when the number of antennas increases, thus leading to a specific subspace structure for the received data vectors that depend on these channels \cite{mueller:14} and 2) the channel covariance matrix exhibits a low-rankness property whenever the multipath impinging on the MIMO array spans a finite angular spread  \cite{yin:13,adhikary:13,yin:14}.
The blind signal subspace estimation in \cite{mueller:14} capitalizes on the first property. The second property has been utilized in \cite{yin:13,adhikary:13,yin:14,yin:14a,adhikary:13a}, assuming the knowledge of the long-term channel covariance matrices.
While the exploitation of the two properties {\em individually} has given rise to a set of distinct original decontamination approaches, in this work we will exploit these two key features in a {\em combined} manner. Doing so we can propose a novel approach towards mitigating pilot contamination that exhibits much higher levels of robustness.

More specifically, in \cite{cottatellucci:13,mueller:14},
the pairwise channel orthogonality property allows to blindly estimate the user-of-interest channel subspace and discriminate between user-of-interest signals and interference based on the channel powers. In practice, decontamination occurs via a
projection driven by the channel amplitudes. This approach works well within the constraint that the interference channel is received with a power level sufficiently lower than that of the desired channel, a condition hard to guarantee for some edge-of-cell users.

In a way completely different from \cite{cottatellucci:13,mueller:14}, another approach based on a linear minimum mean squared error (MMSE) estimator is adopted in \cite{yin:13} to estimate the channel of interest via projection of the received signals onto the user-of-interest subspace. This subspace, identified by a channel covariance matrix (a long-term one, as opposed to the instantaneous signal correlation matrix of \cite{cottatellucci:13,mueller:14}), is related to the angular spread of the signal of interest \cite{yin:13} and enables to annihilate the interference from users with non-overlapping domains of multipath angles-of-arrival (AoA). Interestingly, this latter approach makes no assumption on received signal amplitudes and can also discriminate against interfering users that are received with similar or even higher powers. Yet, the approach fails to decontaminate pilots when propagation scattering creates large angle spread, causing spatial overlapping among desired and interference channels.

In this paper, we point out that the strengths of these two previously unrelated estimation methods are strongly complementary, offering a unique opportunity for developing robust channel estimation schemes.
Thus, we aim to properly merge the two projections in complementary domains while keeping the individual benefits.
In fact, we propose a family of algorithms striking various performance/complexity trade-offs.

We start by presenting a first scheme named ``covariance-aided amplitude based projection" that effectively combines projections in the angular and amplitude domains and exhibits robustness to interference power/angles overlapping conditions.
We present an asymptotic analysis which reveals the conditions under which the channel estimation error due to pilot contamination
and noise can be made to vanish. An intuitive physical interpretation of this condition for a Uniform Linear Array (ULA) is given in the form of the residual interference channel energy contained in the multipath components that overlap in angle with those of the desired channel. Although the physical explanation is given for the ULA example, the general principle apply to other antenna placement topologies.

The obtained condition  for decontamination is in general less restrictive than the condition required by previous MMSE and the amplitude projection-based methods taken separately to achieve complete removal of pilot contamination.

We then propose two low-complexity alternative schemes called ``subspace and amplitude based projection" and ``MMSE + amplitude based projection" respectively.
Such schemes achieve different complexity-performance trade-off at moderate number of antennas. Specifically, the ``subspace and amplitude based projection" can be shown to reach asymptotic (in the number of antennas) decontamination result under the same channel topology conditions as the first scheme.

More specifically, our contributions are as follows:
\begin{itemize}
  \item We put forward a modification of the known method of amplitude based projection, with increased robustness.
  \item We propose a spatial filter which helps bring down the power of interference while preserving the signal of interest. With this spatial filter, we present a novel channel estimation scheme called ``covariance-aided amplitude based projection". It combines the merits of linear MMSE estimator and amplitude based projection method, yet can be shown to have significant gains over these known schemes.
  \item We analyze the asymptotic performance of this proposed method and provide weaker condition compared to the previous methods where the estimation error of the proposed method goes to zero asymptotically in the limit of large number of antennas and data symbols. The asymptotic analysis relies on mild technical conditions such as uniformly boundedness of the spectral norm of channel covariance. 
  \item As the uniformly boundedness of the largest eigenvalue of channel covariance was reported to be useful in previous works (such as \cite{wagner2012}) but not formally analyzed, we identify in the case of ULA a sufficient propagation condition under which the uniformly bounded spectral norm of channel covariance is satisfied exactly.
  \item Finally we propose two low-complexity alternatives of the first method.
      An asymptotic performance characterization is also given.

\end{itemize}

The paper is organized as follows: In section \ref{modeling} we introduce the system model. Section \ref{sec:mmse} is a brief review of MMSE channel estimator and its asymptotic performance. In section \ref{sec:blind} we briefly recall the amplitude based projection of \cite{cottatellucci:13,mueller:14}, and we propose a first improvement of the method. Then we present the novel covariance-aided amplitude based projection in section \ref{subsec:singleUser} for the setting of single user per cell, and the asymptotic performance analysis of this method is shown in section \ref{sec:PerformanceAnalysisCA}. Section \ref{subsec:multiuser} presents a generalization of the proposed scheme to multi-user per cell scenario. In section \ref{sec:lowComplexityAlter} we propose two low-complexity alternatives of our previous method and similar asymptotic results on the system performance are given. Section \ref{sec:numericalResult} shows numerical results. Finally section \ref{conclusion} concludes the paper.

The notations adopted in the paper are as follows. We use boldface to denote matrices and vectors. Specifically, ${\mathbf{I}}_M$ denotes the $M \times M$ identity matrix. ${({\mathbf{X}})^T}$, ${({\mathbf{X}})^*}$, and ${({\mathbf{X}})^H}$ denote the transpose, conjugate, and conjugate transpose of a matrix ${\mathbf{X}}$ respectively. ${({\mathbf{X}})^\dag}$ is the Moore-Penrose pseudoinverse of ${\mathbf{X}}$. $\operatorname{tr}\left\{ \cdot \right\}$ denotes the trace of a square matrix. ${\left\| \cdot \right\|_2}$ denotes the $\ell^2$ norm of a vector when the argument is a vector, and the spectral norm when the argument is a matrix. In particular, if $\mA$ is a Hermitian matrix, ${\left\| \mA \right\|_2}$ is the largest eigenvalue of $\mA$. We index the eigenvalues of $\mA$ in non-increasing order and denote the $i$-th eigenvalue of $\mA$ by $\lambda_i\{ \mA \}$ and its corresponding eigenvector by $\ve_i\{\mA\}$.
${\left\| \cdot \right\|_F}$ stands for the Frobenius norm.
$\mathbb{E}\left\{ \cdot \right\}$ denotes the expectation.
The Kronecker product of two matrices ${\bf{X}}$ and ${\bf{Y}}$ is denoted by ${\bf{X}}\otimes{\bf{Y}}$. $\operatorname{vec} (\mX)$ is the vectorization of the matrix $\mX$.
${\mathop{\rm diag}\nolimits} \{ {\bf{a_1,...,a_N}}\}$ denotes a diagonal matrix or a block diagonal matrix with $\bf{a_1,...,a_N}$ at the main diagonal. $\triangleq$ is used for definition.

\section{Signal and Channel Models}\label{modeling}
We consider a network of $L$ time-synchronized\footnote{Note that assuming synchronization between uplink pilots provides a worst case scenario from a pilot contamination point of view, since any lack of synchronization will tend to statistically decorrelate the pilots. Furthermore, the main methods that we propose in this paper, i.e., the covariance-aided amplitude based projection and the subspace and amplitude based projection do not rely on accurate time synchronization.} cells, with full spectrum reuse. Each base station (BS) is equipped with $M$ antennas. 
There are $K$ single-antenna users in each cell simultaneously served by their base station. The cellular network operates in time-division duplexing (TDD) mode, and due to channel reciprocity, the downlink channel is obtained at the BS by uplink training signal and data signal. Each base station estimates the channels of its $K$ users during a coherence time interval. The pilot sequences inside each cell are assumed orthogonal to each other in order to avoid intra-cell interference. However the same pilot pool is reused in other cells, giving rise to pilot contamination problem. The pilot sequence assigned to the $k$-th user in a certain cell is denoted by
 \begin{equation}\label{Eq:pilots}
    {{\mathbf{s}}_k} = {[\begin{array}{*{20}{c}}
  {{s_{k1}}}&{{s_{k2}}}& \cdots &{{s_{k\tau }}}
\end{array}]^T},
 \end{equation}
where $\tau$ is the length of a pilot sequence. Without loss of generality we assume unitary average power of pilot symbols:
\begin{equation*}
{\mathbf{s}}^T_{k_1} {\mathbf{s}}^*_{k_2} =
\left\{\begin{matrix}
0, &{k_1} \neq {k_2}\\
\tau , &{k_1} = {k_2}
\end{matrix}\right..
\end{equation*}

The $M \times 1$ channel vector between the $k$-th user located in the $l$-th cell and the $j$-th base station is denoted by $\vh_{lk}^{(j)}$.
The following classical multipath channel model \cite{molisch:10} is adopted:
\begin{equation}\label{Eq:simpleChanModel}
{{\mathbf{h}}_{lk}^{(j)}} = \frac{\beta_{lk}^{(j)}}{\sqrt{P}} \sum\limits_{p = 1}^P {{\mathbf{a}}({\theta _{lkp}^{(j)}}){e^{i\varphi_{lkp}^{(j)}}}},
\end{equation}
where $P$ is an arbitrary large number of i.i.d. paths, and $e^{i\varphi_{lkp}^{(j)}}$ is their i.i.d. random phase, which is independent over channel indices $l$, $k$, $j$, and path index $p$. ${\mathbf{a}}({\theta})$ is the steering (or phase response) vector by the array to a path originating from the angle of arrival $\theta$:
\begin{equation}\label{Eq:steeVec}
{\mathbf{a}}({\theta }) \triangleq \left[ {\begin{array}{*{20}{c}}
  1 \\
  {{e^{ - j2\pi \frac{D}{\lambda }\cos ({\theta})}}} \\
   \vdots  \\
  {{e^{ - j2\pi \frac{{(M - 1)D}}{\lambda }\cos ({\theta})}}}
\end{array}} \right],
\end{equation}
where $\lambda$ is the signal wavelength and $D$ is the antenna spacing which is assumed fixed. Note that we can limit $\theta$ to $\theta \in [0,\pi]$ because any $\theta \in [-\pi,0)$ can be replaced by $-\theta$ giving the same steering vector.
${\beta_{lk}^{(j)}}$ is the path-loss coefficient
\begin{equation}\label{Eq:pathloss}
{\beta_{lk}^{(j)}} = \sqrt{\frac{\alpha }{{{d_{lk}^{(j)}}^\gamma }}},
\end{equation}
in which $\gamma$ is the path-loss exponent, $d_{lk}^{(j)}$ is the geographical distance between the user and the $j$-th base station, and $\alpha$ is a constant. Note
that the model is shown for a ULA example for ease of exposition. Under this model, the covariance matrix can be shown asymptotically to have low rank, as long as the AoA support is bounded and strictly smaller than $[0, \pi]$. However, several other channel models also exhibit similar low-rank property \cite{yin:14}, which is the essential characteristic exploited by the MMSE estimator. Hence our approach is not dependent on the use of the one ring model above described.
In fact, our main results, namely Theorem \ref{theoStatis}, as well as the general principle carry to other channel models and antenna placement topologies.

We define
\begin{equation}
{\mathbf{H}}_{l}^{(j)} \triangleq
\begin{bmatrix}
\vh_{l1}^{(j)} & \vh_{l2}^{(j)} & \cdots & \vh_{lK}^{(j)}
\end{bmatrix},
\end{equation}
and the pilot matrix
\begin{equation}
\mS \triangleq \begin{bmatrix}
\vs_1  & \vs_2  & \cdots  & \vs_K
\end{bmatrix}^T.
\end{equation}

During the training phase, the received signal at the base station $j$ is
\begin{equation}\label{Eq:train}
{{\mathbf{Y}^{(j)}}} = \sum\limits_{l = 1}^L {{{\mathbf{H}}_{l}^{(j)}}{\mathbf{S}}+ {{\mathbf{N}^{(j)}}}},
\end{equation}
where ${{\mathbf{N}^{(j)}}} \in {\mathbb{C}^{M \times \tau }}$ is the spatially and temporally white additive Gaussian noise (AWGN) with zero-mean and element-wise variance ${\sigma _n^2}$.
Then, during the uplink data transmission phase, each user transmits $C$ data symbols. The received data signal at base station $j$ is given by:
\begin{equation}\label{Eq:ULdata}
\mW^{(j)} = \sum\limits_{l = 1}^L {{{\mathbf{H}}_{l}^{(j)}} {\mathbf{X}_l}+ {\mZ}^{(j)}},
\end{equation}
where $\mX_l \in {\mathbb{C}^{K \times C }}$ is the matrix of transmitted symbols of all users in the $l$-th cell. The symbols are i.i.d. with zero-mean and unit average element-wise variance. $\mZ^{(j)} \in {\mathbb{C}^{M \times C }}$ is the AWGN noise with zero-mean and element-wise variance ${\sigma _n^2}$. Note that the block fading channel is constant during the transmission for the $\tau$ pilot symbols and the $C$ data symbols.

\section{MMSE channel estimation}\label{sec:mmse}
We briefly recall the MMSE channel estimator in a multi-cell setting with single-user per cell. Without loss of generality, we assume cell $j$ is the target cell, and ${\vh}^{(j)}_j \in {\mathbb{C}^{M \times 1 }}$ is the desired channel, while ${\vh}^{(j)}_l \in {\mathbb{C}^{M \times 1 }}, \forall l \neq j$ are the interference channels.
We rewrite (\ref{Eq:train}) in a vectorized form,
\begin{equation}\label{Eq:vecMod}
{\mathbf{y}^{(j)}} = {\mathbf{\overline S}}\sum\limits_{l = 1}^L{{{\vh}^{(j)}_l}} + {\mathbf{n}^{(j)}},
\end{equation}
where ${\mathbf{y}^{(j)}} = \operatorname{vec} ({{\mathbf{Y}^{(j)}}})$, ${{\mathbf{n}^{(j)}}} = \operatorname{vec} ({{\mathbf{N}^{(j)}}})$.
A pilot sequence $\vs$ is shared by all users. The pilot matrix ${\mathbf{\overline S}}$ is given by
\begin{equation}\label{Eq:pilotMat}
  {\mathbf{\overline S}} \triangleq {{{\mathbf{s}}} \otimes {{\mathbf{I}}_M}}.
\end{equation}
We define the channel covariance matrices
\begin{align}
{{\mathbf{R}}^{(j)}_{l}} \triangleq \mathbb{E}\{ {{{{\mathbf{h}}}^{(j)}_{l}}{{\mathbf{h}}}_{l}^{{(j)}H}} \} \in {\mathbb{C}^{M \times M }}, l = 1, \ldots, L,
\end{align}
where the expectation is taken over channel realizations.

A linear MMSE estimator for $\vh^{(j)}_j$ is given by
\begin{align}\label{Eq:MMSE_2}
\widehat{{{\vh}}}_j^{(j){\text{MMSE}}} = {\mR}^{(j)}_j  \left( \tau \sum_{l=1}^{L} {\mR}^{(j)}_l + \sigma_n^2 \mI_{M} \right)^{-1} \overline{\mS}^H \vy^{(j)}.
\end{align}
As shown in previous works \cite{yin:13, yin:14}, for a base station equipped with a ULA, the above MMSE estimator can fully eliminate the effects of interfering channels when $M\rightarrow \infty$, under a specific ``non-overlap" condition on the distributions of multipath AoAs for the desired and interference channels. This condition is formalized as follows. Assume the user in cell $j$ is our target (desired) user.
Denote the angular support of the desired channel as $\Phi_d$, (i.e., the probability density function (PDF) $p_d(\theta)$ of the AoA of the desired channel $\vh_{j}^{(j)}$ satisfies $p_d(\theta)>0$ if $\theta \in \Phi_d$ and $p_d(\theta)=0$ if $\theta \notin \Phi_d$) and similarly the union of the angular supports of all interference channels $\vh_{l}^{(j)} (l \neq j)$ as $\Phi_i$. If $\Phi_d \cap \Phi_i = \emptyset$, then, as $M\rightarrow \infty$, (\ref{Eq:MMSE_2}) converges to an interference-free estimate. In practice the ``non-overlap" condition is hard to guarantee and the finite-$M$ performance of the MMSE scheme depends on angular spread and user location, although the latter can be shaped via the use of so-called coordinated pilot assignment (CPA) \cite{yin:13}.

\section{Amplitude based projection}\label{sec:blind}
Interestingly, angle is not the only domain where interference can be discriminated upon, as revealed from a completely different approach to pilot decontamination  \cite{cottatellucci:13,mueller:14}. In that approach the empirical instantaneous covariance matrix built from the received data (\ref{Eq:ULdata}) is exploited, in contrast to the use of long-term covariance matrices in (\ref{Eq:MMSE_2}). Assume cell $j$ is our target cell and each cell has $K$ users.
The eigenvalue decomposition (EVD) of $\mW^{(j)} \mW^{(j)H}/C$ is written as
\begin{equation}\label{Eq:EVD_W}
\frac{1}{C} \mW^{(j)} \mW^{(j)H} = \mU^{(j)} \mathbf{\Lambda}^{(j)} \mU^{(j)H},
\end{equation}
where $\mU^{(j)} \in \mathbb{C}^{M \times M} = \begin{bmatrix}
\vu^{(j)}_{1} & \vu^{(j)}_{2} & \cdots & \vu^{(j)}_{M}
\end{bmatrix}$ is a unitary matrix and $\mathbf{\Lambda}^{(j)} = {\mathop{\rm diag}\nolimits} \{\lambda^{(j)}_1, \cdots, \lambda^{(j)}_M\}$ with its diagonal entries sorted in a non-increasing order. By extracting the first $K$ columns of $\mU^{(j)}$, i.e., the eigenvectors corresponding to the strongest $K$ eigenvalues, we obtain an orthogonal basis
\begin{equation}\label{Eq:basis}
\mE^{(j)} \triangleq \begin{bmatrix} \vu^{(j)}_{1} & \vu^{(j)}_{2} & \cdots & \vu^{(j)}_{K}\end{bmatrix} \in \mathbb{C}^{M \times K}.
\end{equation}
The basic idea in \cite{mueller:14, cottatellucci:13} is
to use the orthogonal basis $\mE^{(j)}$ as an estimate for the span of ${\mH}^{(j)}_j$, which includes all desired user channels in cell $j$. Then, by projecting the received signal onto the subspace spanned by $\mE^{(j)}$, most of the signal of interest is preserved. In contrast, the interference signal is canceled out thanks to the asymptotic property that the user channels are pairwise orthogonal as the number of antennas tends to infinity.
Thus after the above mentioned projection, the estimate of the multi-user channel ${\mH}^{(j)}_j$ is given by:
\begin{align}\label{Eq:pureBlind}
\widehat{{\mH}}_j^{(j) \text{AM}} &= \frac{1}{\tau} \mE^{(j)} \mathbf{E}^{(j)H} \mathbf{Y}^{(j)} \mS^H.
\end{align}
Note here that interference and desired channel directions are discriminated on the basis of channel amplitudes and not AoA, hence the estimate is labeled ``AM" for ``Amplitude". As a way to guarantee an asymptotic separation between the signal of interest and the interference in terms of power,
it has been suggested to introduce power control in the network \cite{mueller:14, cottatellucci:13}.
{\remark{Generalized amplitude projection}\label{remark1}}

As shown in \cite{mueller:14, cottatellucci:13}, the above method works well when the desired channels and interference channels are separable in power domain, i.e., the instantaneous powers of any desired channels are higher than that of any interference channels. In practice however, this assumption is not always guaranteed. For a finite number of antennas, the short-term fading realization can cause the interference subspace to spill over the desired one. An enhanced version can somewhat mitigate this problem by considering a generalized amplitude based projection. This consists in selecting a possibly larger number ($\kappa^{(j)}$) of dominant eigenvectors to form $\mE^{(j)}$, where $\kappa^{(j)}$ is the number of eigenvalues in $\mathbf{\Lambda}^{(j)}$ that are greater than $\mu \lambda^{(j)}_K$. $\mu$ is a design parameter that satisfies $0 \leq \mu < 1$.
See section \ref{sec:numericalResult} for details on the choice of $\mu$.

\section{Covariance-aided amplitude based projection}\label{sec:cova-aidedBlind}
Note that both previous methods, while being able to tackle pilot contamination in quite different ways, perform well only in some restricted user/channel topologies.
For a ULA base station, the MMSE method leads to interference-free channel estimates under the strict requirement that the desired and interference channel do not overlap in their AoA regions.
While the amplitude based projection requires that
no interference channel power exceeds that of a desired channel to achieve a similar result. Unfortunately, due to the random user location and scattering effects, it is quite unlikely to achieve these conditions at all times.
As a result, by combining the useful properties of both the MMSE and the amplitude projection method, we propose below novel estimation methods that will lead to enhanced robustness in a realistic cellular scenario.

\subsection{Single user per cell}\label{subsec:singleUser}
For ease of exposition we first consider a simplified scenario where intra-cell interference is ignored by assuming that each cell has only one user, i.e., $K=1$. The users in different cells share the same pilot sequence $\vs$. Then, with proper modifications, we will generalize this method to the setting of multiple users per cell in section \ref{subsec:multiuser}.

The objective is to combine long-term statistics which include spatial distribution information together with short-term empirical covariance which contains instantaneous amplitude and direction channel information. Hence, a spatial distribution filter can be associated to an instantaneous projection operator to help discriminate against any interference terms whose spatial directions live in a subspace orthogonal to that of the desired channel. The intuition is that such a spatial filter may bring the residual interference to a level that is acceptable to the instantaneous projection-based channel estimator.

In order to carry out the above intuition, we introduce a long-term statistical filter ${{\bf{\Xi}}_{j}}$, which is based on channel covariance matrices in a way similar to that used by the MMSE filter in (\ref{Eq:MMSE_2}).
\begin{equation}\label{Eq:Xij1}
{{\bf{\Xi}}_{j}} = {\left( {\sum\limits_{l = 1}^L {{{\bf{R}}_{l}^{(j)}}}  + \sigma _n^2{{\bf{I}}_M}} \right)^{ - 1}} {{\bf{R}}_{j}^{(j)}}.
\end{equation}

Note that the linear filter ${{\bf{\Xi}}_{j}}$ allows to discriminate against the interference in angular domain by projecting away from multipath AoAs that are occupied by interference. Note also that the choice of spatial filter ${{\bf{\Xi}}_{j}}$ is justified from the fact that the full information of desired channel ${{\bf{h}}_{j}^{(j)}}$ is preserved, as ${{\bf{h}}_{j}^{(j)}}$ lies in the signal space of ${{\bf{R}}_{j}^{(j)}}$. In fact, the desired channel is recoverable using another linear transformation ${{\bf{\Xi}}_{j}}'$:
\begin{equation}\label{Eq:Xi'}
{{\bf{\Xi}}_{j}}' \triangleq {{\bf{R}}_{j}^{(j)^\dag}} {\left( {\sum\limits_{l = 1}^L {{{\bf{R}}_{l}^{(j)}}}  + \sigma _n^2{{\bf{I}}_M}} \right)},
\end{equation}
as can be seen from the following equality
\begin{equation}
{\bf{\Xi}}_{j}' {\bf{\Xi}}_{j}  {\bf{{{h}}}}_{j}^{(j)} = {{\bf{R}}_{j}^{(j)\dag}} {{\bf{R}}_{j}^{(j)}} {\bf{{{h}}}}_{j}^{(j)} =  \mV_{j} \mV_{j}^H {\bf{{{h}}}}_{j}^{(j)} = {\bf{{{h}}}}_{j}^{(j)},
\end{equation}
where the columns of $\mV_{j}$ are the eigenvectors of ${\mathbf{R}}^{(j)}_{j}$ corresponding to non-zero eigenvalues. 

The spatial filter is applied to the received data signal at base station $j$ as
\begin{equation}\label{Eq:Wtilde}
{\bf{\widetilde W}}_j \triangleq {{\bf{\Xi}}_{j}}{\bf{W}}^{(j)}.
\end{equation}
The amplitude-based method as shown in section \ref{sec:blind} can now be applied on the filtered received data to get rid of the residual interference.
Take the eigenvector corresponding to the largest eigenvalue of the matrix ${\bf{\widetilde W}}_j{{\bf{\widetilde W}}_j^{H}}/C$:
\begin{equation}\label{Eq:tilde_uj1}
\widetilde{\bf{u}}_{j1} = \ve_1\{ \frac{1}{C} {\bf{\widetilde W}}_j{{\bf{\widetilde W}}_j^{H}} \}.
\end{equation}
Hence $\widetilde{\bf{u}}_{j1}$ can be considered as an estimate of the direction of the vector ${{\bf{\Xi}}_{j}} \vh_{j}^{(j)}$.

We then cancel the effect of the pre-multiplicative matrix ${{\bf{\Xi}}_{j}}$ using ${\bf{\Xi}}_{j}'$ in (\ref{Eq:Xi'}), and we obtain an estimate of the direction of the channel vector ${{\bf{h}}_{j}^{(j)}}$ as follows:
\begin{equation}\label{Eq:u_bar}
{\overline{\bf{u}}}_{j1} = \frac{{\bf{\Xi}}_{j}' \widetilde{\bf{u}}_{j1}}{\left\|{\bf{\Xi}}_{j}' \widetilde{\bf{u}}_{j1}\right\|_2}.
\end{equation}

Finally, the phase and amplitude ambiguities of the desired channel can be resolved by projecting the LS estimate onto the subspace spanned by ${\overline{\bf{u}}}_{j1}$:
\begin{equation}\label{Eq:CAProj}
{\bf{\widehat {{h}}}}_{j}^{(j)\text{CA}} = \frac{1}{\tau} \overline{\bf{u}}_{j1} \overline{\bf{u}}^{H}_{j1} \mathbf{Y}^{(j)} \vs^*,
\end{equation}
where the superscript ``CA" denotes the covariance-aided amplitude domain projection.

The algorithm is summarized below:
\begin{algorithm}
\caption{Covariance-aided Amplitude based Projection}
\begin{algorithmic}[1]\label{Alg:CA_estimator}
\STATE{Take the first eigenvector of ${\bf{\widetilde W}}_j{{\bf{\widetilde W}}_j^{H}}/C$ as in (\ref{Eq:tilde_uj1}), with ${\bf{\widetilde W}}_j$ being the filtered data signal.}

\STATE{Reverse the effect of the spatial filter using (\ref{Eq:u_bar}).}

\STATE{Resolve the phase and amplitude ambiguities by (\ref{Eq:CAProj}).
}
\end{algorithmic}
\end{algorithm}

The complexity of this proposed estimation scheme is briefly evaluated.

We note that the computation of the matrix inversions in (\ref{Eq:Xij1}) has a complexity order of $O(M^{2.37})$. However, these computations are performed in a preamble phase and their cost is negligible under the underlying assumption of channel stationarity implicitly made in this article. In practical systems, the matrix inversion in (\ref{Eq:Xij1}) is performed when the channel statistics are updated. Since the channel statistics are typically updated in a time scale much larger than the channel coherence time, i.e., the time scale for the applicability of Algorithm \ref{Alg:CA_estimator}, then their computational cost is negligible. 
Therefore, we can focus on the complexity of Algorithm \ref{Alg:CA_estimator} only.

In step 1, the spatial filtering of the data signals in (\ref{Eq:Wtilde}) and the computation of the covariance matrix ${\bf{\widetilde W}}_j{{\bf{\widetilde W}}_j^{H}}$ is performed along with the computation of the dominant eigenvector of an $M \times M$ matrix as in (\ref{Eq:tilde_uj1}). The former computation has a complexity order $O(CM^2)$ while, by applying the classical power method, the computation of the dominant vector has a complexity order $O(M^2)$. Both step 2 and step 3 require multiplications of matrices by $M$-dimensional vectors and thus both have a complexity order $O(M^2).$ Then, the global complexity of the algorithm is dominated by the complexity of step 1, which is $O(CM^2)$.

The ability for the above estimator to combine the advantages of the previously known angle and amplitude projection based estimators is now analyzed theoretically. In particular we are interested in the conditions under which full pilot decontamination can be achieved asymptotically in the limit of the number of antennas $M$ and data symbols $C$.
In order to facilitate the analysis, we introduce the following condition.

\textit{Condition C1}: The spectral norm of ${{\bf{R}}_{l}^{(j)}}$ is uniformly bounded:
\begin{equation}\label{Eq:condition1}
\forall M \in \mathbb{Z}^+ \text{ and } \forall l \in \{1,\ldots, L\}, \exists \zeta, \text{s.t.} \left\|{{\bf{R}}_{l}^{(j)}}\right\|_2 < \zeta,
\end{equation}
where $\mathbb{Z}^+$ is the set of positive integers and $\zeta$ is a constant.

Condition C1 can be interpreted as describing all the scenarios in which the channel energy is spread over a subspace whose dimension grows with $M$. Note that the same assumption can be found in some other papers,
e.g., \cite{wagner2012}. 
The corresponding physical condition is now investigated for the case of a ULA with a typical antenna spacing $D$ (less than or equal to half wavelength).
\begin{proposition}\label{propBoundedness}
Let $\Phi$ be the AoA support of a certain user. Let $p(\theta)$ be the probability density function of AoA of that user. If $p(\theta)$ is uniformly bounded, i.e., $p(\theta) < +\infty, \forall \theta \in \Phi$, and $\Phi$ lies in a closed interval that does not include the parallel directions with respect to the array , i.e., $0, \pi \notin \Phi$, then, the spectral norm of the user's covariance ${\bf{R}}$ is uniformly bounded.
\end{proposition}
\begin{proof}
\quad \emph{Proof:} See Appendix \ref{proof:propBoundedness}.
\end{proof}

Note that this result is hinted upon \cite{adhikary:13} by resorting to approximation of ${\bf{R}}$ by a circulant matrix. Our Proposition \ref{propBoundedness} here gives a formal proof of the previous approximated result.

As another interpretation of Condition C1, it is worth noting that when this condition is not satisfied, there is no guarantee that the asymptotic pairwise orthogonality of different users' channels holds. In other words, the quantity ${{\bf{h}}_{j}^{(j)H}} {{\bf{h}}_{l}^{(j)}}/M, l \neq j$ may not converge to zero, which is an adverse condition for all massive MIMO methods. However, our proposed methods still have significant performance gains under this adverse circumstance. Moreover, C1 is a sufficient condition and we believe it can be weakened.

\subsection{Asymptotic performance of the proposed CA estimator}\label{sec:PerformanceAnalysisCA}
We now look into the performance analysis of the proposed estimation scheme. Let us define
\begin{align}\label{Eq:Alphal}
\alpha_{l}^{(j)} \triangleq \mathop {\lim }\limits_{M \to \infty } \frac{1}{M} \operatorname{tr} \{ {\bf{\Xi}}_j {{\bf{R}}_{l}^{(j)}} {\bf{\Xi}}_j^H \}, \forall l = 1,\ldots, L.
\end{align}
\begin{theorem}\label{theoStatis}
Given Condition C1,
if the following inequality holds true
\begin{equation}\label{Eq:alphaj>l}
\alpha_{j}^{(j)} > \alpha_{l}^{(j)}, \forall l \neq j,
\end{equation}
then, the estimation error of (\ref{Eq:CAProj}) vanishes, i.e.,
\begin{equation}\label{Eq:limit_statis_multi-int}
\mathop {\lim }\limits_{M, C \to \infty } \frac{\left\|{\bf{\widehat {{h}}}}_{j}^{(j)\text{CA}} - {{\vh}}_{j}^{(j)} \right\|_2^2}{\left\|{{\vh}}_{j}^{(j)}\right\|_2^2} = 0.
\end{equation}
\end{theorem}\begin{proof}
\quad \emph{Proof:}
For the sake of notational convenience, in this proof we assume the user in cell $j$ is the target user and thus drop the superscript $(j)$. The desired channel is denoted by $\vh_j = {\vh}_{j}^{(j)}$ and the interference channels are $\vh_l = {\vh}_{l}^{(j)}, l\neq j$.
Since $\vh_l, l=1,\ldots, L,$ is considered as $M \times 1$ complex Gaussian with the spatial correlation matrices $\mR_l = \mathbb{E} \{\vh_l \vh_l^H\}$, the channels can be factorized as \cite{paulraj2003}
\begin{equation}\label{Eq:spaWhi}
{{\mathbf{h}}_{l}} = {\mathbf{R}}_{l}^{1/2}{{\mathbf{h}}_{\text{W}}}_{l}, l=1,\ldots,L,
\end{equation}
where ${{\mathbf{h}}_{\text{W}}}_{l} \sim \mathcal{CN}({\mathbf{0}},{\mathbf{I}_{M}})$, is an i.i.d. $M\times1$ Gaussian vector with unit variance.
We build the proof of Theorem \ref{theoStatis} on the general correlation model (\ref{Eq:spaWhi}).
The proof consists in three parts, corresponding to the three steps in Algorithm \ref{Alg:CA_estimator} respectively. More specifically, Lemma \ref{lemmaEigenvectorFiltered} (and the intermediate results towards Lemma \ref{lemmaEigenvectorFiltered}) is the first part of the proof. It shows that $\widetilde{\bf{u}}_{j1}$ aligns asymptotically with the direction of the filtered channel vector $\underline{\vh}_j = {\bf{\Xi}}_j {\mathbf{h}}_{j}$. The second part of the proof is provided in Lemma \ref{lemmaEigenvectorOrig}, which proves that after canceling the effect of the spatial filter using ${\bf{\Xi}}'_j$, we obtain the direction of the true channel ${\mathbf{h}}_{j}$ in ${\overline{\bf{u}}}_{j1}$. The final part of the proof shows that by projecting the LS estimate onto the subspace of ${\overline{\bf{u}}}_{j1}$, we resolve the phase and amplitude of the true channel.

\begin{lemma}\label{lemmaEigenvectorFiltered}
Given Condition C1, if $\alpha_{j}^{(j)} > \alpha_{l}^{(j)}, \forall l \neq j$, then there exists a unique $0 \leq \phi < 2\pi$, such that
\begin{align}\label{Eq:lemmaEigenvectorFiltered}
\mathop {\lim }\limits_{M, C \to \infty } {\left\| \frac{\underline{\vh}_j}{\left\| \underline{\vh}_j \right\|_2} - \widetilde{\bf{u}}_{j1} e^{j\phi} \right\|_2 } = 0.
\end{align}
where
$\underline{\vh}_l \triangleq {\bf{\Xi}}_j {\mathbf{h}}_{l}, l = 1, \ldots, L$.
\end{lemma}
\begin{proof}
\quad \emph{Proof:}
The proof of Lemma \ref{lemmaEigenvectorFiltered} relies on several intermediate results, namely Lemma \ref{lemmaSpectralNormPerturbation} - Lemma \ref{lemmaEigenvector}.
\begin{lemma}\label{lemmaSpectralNormPerturbation}
Under Condition C1,
the spectral norm of ${\bf{\Xi}}_j{\bf{\Xi}}_j^H$ satisfies:
\begin{align}
\mathop {\lim }\limits_{M \to \infty } \frac{1}{M} \left\| {\bf{\Xi}}_j{\bf{\Xi}}_j^H \right\|_2 = 0.
\end{align}

\end{lemma}
\begin{proof}
\quad \emph{Proof:} See Appendix \ref{proof:lemmaSpectralNormPerturbation}.
\end{proof}
Lemma \ref{lemmaSpectralNormPerturbation} indicates that the spectral norm of the covariance of the noise (after multiplying ${\bf{\Xi}}_j$) is bounded and does not scale with $M$. This conclusion will be exploited when we prove in Lemma \ref{lemmaEigenvector} that the impact of noise on the dominant eigenvector/eigenvalue vanishes.

\begin{lemma}\label{lemmaQuadratic}
\cite{evans2000} Let $\mA_M$ be a deterministic $M \times M$ complex matrix with uniformly bounded spectral radius for all $M$. Let $\vq = \frac{1}{\sqrt(M)}[q_1, \cdots, q_M]^T$ where $q_i, \forall i = 1, \cdots, M$ is i.i.d. complex random variable with zero mean, unit variance, and finite eighth moment. Let $\vr$ be a similar vector independent of $\vq$. Then as $M \rightarrow \infty$,
\begin{equation}\label{Eq:qadratic1}
\vq^H \mA_M \vq {\xrightarrow{\text{a.s.}}} \frac{1}{M} \operatorname{tr}\{ \mA_M\},
\end{equation}
and
\begin{equation}\label{Eq:qadratic2}
\vq^H \mA_M \vr {\xrightarrow{\text{a.s.}}} 0,
\end{equation}
where ${\xrightarrow{\text{a.s.}}}$ denotes almost sure convergence.
\end{lemma}
Note that in this paper, the condition on the finite eighth moment always holds, as when we apply Lemma \ref{lemmaQuadratic}, the components of the vector of interest are i.i.d. complex Gaussian variables. It is well known that a complex Gaussian variable with zero mean, unit variance has finite eighth moment.

\begin{lemma}\label{lemmaInnerProd}
Given Condition C1,
\begin{align}
\mathop {\lim }\limits_{M \to \infty } \frac{1}{M} {\underline{\vh}_j^H} \underline{\vh}_l & = 0, \forall l \neq j \label{Eq:hjhl_underline}\\
\mathop {\lim }\limits_{M \to \infty } \frac{1}{M} \underline{\vh}_l^H \underline{\vh}_l & = \alpha_l, l = 1, \ldots, L. \label{Eq:hlhl_underline}
\end{align}

\end{lemma}

\begin{proof}
\quad \emph{Proof:} See Appendix \ref{proof:lemmaInnerProd}.
\end{proof}

\begin{lemma}\label{lemmaEigenvector}
When Condition C1 is satisfied,
\begin{align}\label{Eq:lemmaEigenvector}
\mathop {\lim }\limits_{M,C \to \infty } \left\| \frac{ {\bf{\widetilde W}}_j{{\bf{\widetilde W}}_j^{H}} }{MC} \frac{\underline{\vh}_j}{\left\|\underline{\vh}_j\right\|_2} - \alpha_j \frac{\underline{\vh}_j}{\left\|\underline{\vh}_j\right\|_2} \right\|_2 &= 0,
\end{align}
\end{lemma}
\begin{proof}
\quad \emph{Proof:} See Appendix \ref{proof:lemmaEigenvector}.
\end{proof}
Lemma \ref{lemmaEigenvector} proves that as $M,C \rightarrow \infty$, $\alpha_j$ is an asymptotic eigenvalue of the random matrix ${\bf{\widetilde W}}_j{{\bf{\widetilde W}}_j^{H}}/{MC}$, with its corresponding eigenvector converging to ${\underline{\vh}_j}/{\left\|\underline{\vh}_j\right\|_2}$ up to a random phase.

We now return to the proof of Lemma \ref{lemmaEigenvectorFiltered}. Since $\alpha_j > \alpha_l, \forall l \neq j$, one may readily obtain from Lemma \ref{lemmaEigenvector} and (\ref{Eq:hjhl_underline}) that
\begin{equation}
\mathop {\lim }\limits_{M,C \to \infty } \lambda_1 \left\{ \frac{ {\bf{\widetilde W}}_j{{\bf{\widetilde W}}_j^{H}} }{MC} \right\} = \alpha_j,
\end{equation}
and that there exists a unique $0 \leq \phi < 2\pi$, such that
\begin{equation}\label{Eq:h_udl_converge}
\mathop {\lim }\limits_{M,C \to \infty } {\left\| \frac{\underline{\vh}_j}{\left\| \underline{\vh}_j \right\|_2} - \ve_1\left\{ \frac{ {\bf{\widetilde W}}_j{{\bf{\widetilde W}}_j^{H}} }{MC} \right\} e^{j\phi} \right\|_2 } = 0,
\end{equation}
which completes the proof of Lemma \ref{lemmaEigenvectorFiltered}.
\end{proof}

Now we show the second part of the proof of Theorem \ref{theoStatis}. Note that in this part we make the implicit assumption that the spectral norm of ${\bf{\Xi}}_j'$ satisfies $\left\| {\bf{\Xi}}_j' \right\|_2 < + \infty$. A sufficient (but not necessary) condition of such an assumption is that the spectral norm of $\mR_j^\dag$ is finite.
\begin{lemma}\label{lemmaEigenvectorOrig}
Given (\ref{Eq:lemmaEigenvectorFiltered}), we have
\begin{equation}\label{Eq:h_converge}
\mathop {\lim }\limits_{M,C \to \infty } {\left\| \frac{{\vh}_j}{\left\| {\vh}_j \right\|_2} - \bar{\bf{u}}_{j1} e^{j\phi} \right\|_2 } = 0.
\end{equation}
\end{lemma}
\begin{proof}
\quad \emph{Proof:} See Appendix \ref{proof:lemmaEigenvectorOrig}.
\end{proof}

The final part of the proof of Theorem \ref{theoStatis} can be found in Appendix \ref{proof:theoStatis}, which corresponds to step 3 of Algorithm \ref{Alg:CA_estimator}. The proof shows that projecting the LS estimate onto the subspace of $\bar{\bf{u}}_{j1}$ will lead to noise-free estimate asymptotically as $M, C \rightarrow \infty$. This concludes the proof of Theorem \ref{theoStatis}.
\end{proof}

Interestingly, condition (\ref{Eq:alphaj>l}) in Theorem \ref{theoStatis} can be replaced with
\begin{equation}
\left\|{{\bf{\Xi}}_{j}} {\bf{{{h}}}}_{j}^{(j)} \right\|_2 > \left\| {{\bf{\Xi}}_{j}} {\bf{{{h}}}}_{l}^{(j)}\right\|_2, \forall l \neq j,
\end{equation}
which indicates that under suitable conditions on the spectral norm of channel covariance, after multiplying the filter ${\bf{\Xi}}_j$, if the power of the desired channel is higher than that of interference channel, then, pilot contamination disappears asymptotically, along with noise.

Note that we have so far no assumption on antenna placement in the analysis, other than the requirement for uniformly boundedness of the spectral norm of channel covariance.
In the sequel we look into a specific model of ULA as an example and seek to further understand the physical meaning of the proposed method.

We still assume ${{\mathbf{h}}_{j}^{(j)}}$ is the channel of interest.
Denote its angular support as $\Phi_d$.
Decompose the interference channel ${{\mathbf{h}}_{l}^{(j)}}, \forall l \neq j$, as follows:
\begin{equation}\label{Eq:decomposeInt}
{{\mathbf{h}}_{l}^{(j)}} = {{\mathbf{h}}_{l\text{i}}^{(j)}} + {{\mathbf{h}}_{l\text{o}}^{(j)}},
\end{equation}
where
\begin{align}
{{\mathbf{h}}_{l\text{i}}^{(j)}} &= \frac{{\beta_{l}^{(j)}} }{\sqrt{P}} \sum\limits_{\theta \in \Phi_d} {{\mathbf{a}}({\theta}){e^{i\varphi_{\theta}}}} \\
{{\mathbf{h}}_{l\text{o}}^{(j)}} &= \frac{{\beta_{l}^{(j)}} }{\sqrt{P}} \sum\limits_{\theta \notin \Phi_d} {{\mathbf{a}}({\theta}){e^{i\varphi_{\theta}}}},
\end{align}
which means ${{\mathbf{h}}_{l\text{i}}^{(j)}}$ is the residual multipath component of the interference channel within the AoA region $\Phi_d$ of the desired channel,
while ${{\mathbf{h}}_{l\text{o}}^{(j)}}$ is the multipath component which is outside $\Phi_d$.

\begin{theorem}\label{theoAnglePower}
For a ULA base station, under Condition C1,
if the residual multipath component of the interference channel satisfies:
\begin{equation}
\forall l \neq j, \left\| {\bf{\Xi}}_j {{\mathbf{h}}_{l\text{i}}^{(j)}} \right\|_2 < \left\|{\bf{\Xi}}_j {{\mathbf{h}}_{j}^{(j)}} \right\|_2,
\end{equation}
then, the estimation error of the estimator (\ref{Eq:CAProj}) vanishes:
\begin{equation}\label{Eq:limit_partialOverlap}
\mathop {\lim }\limits_{M, C \to \infty } \frac{\left\|{\bf{\widehat {{h}}}}_{j}^{(j)\text{CA}} - {{\vh}}_{j}^{(j)} \right\|_2^2}{\left\|{{\vh}}_{j}^{(j)}\right\|_2^2} = 0.
\end{equation}
\end{theorem}
\begin{proof}
\quad \emph{Proof:}
See Appendix \ref{proof:theoAnglePower}.
\end{proof}
Theorem \ref{theoAnglePower} further confirms the fact that for a base station equipped with ULA, only the interference multipath components that overlap with those of the desired channel affect the performance of our pilot decontamination method. In other words, the spatial filter ${\bf{\Xi}}_j$ removes the energy located in all interference multipath originating from directions that do not overlap with those of the desired channel. It is then sufficient for the energy of the {\em residual} interference components to be below that of the desired channel to allow for a full decontamination.

\subsection{Generalization to multiple users per cell}\label{subsec:multiuser}
Now we generalize the covariance-aided amplitude based projection into multi-user setting where $K$ users are served simultaneously in each cell. We consider the estimation of user channel ${\vh}_{jk}^{(j)}$ in the reminder of this section. 

Define a matrix ${{\mH}}_{j\backslash k}^{(j)}$ as a sub-matrix of ${{\mH}}_{j}^{(j)}$ after removing its $k$-th column,
\begin{equation}\label{Eq:subChannel}
{{\mH}}_{j\backslash k}^{(j)} \triangleq \begin{bmatrix}{\vh}_{j1}^{(j)} & \cdots & {\vh}_{j(k-1)}^{(j)}& {\vh}_{j(k+1)}^{(j)} & \cdots & {\vh}_{jK}^{(j)} \end{bmatrix}.
\end{equation}
A corresponding estimate of (\ref{Eq:subChannel}), denoted by $\widehat{{\mH}}_{j\backslash k}^{(j)}$, is obtained by removing the $k$-th column of $\widehat{{\mH}}_{j}^{(j)}$, which can be an LS estimate, MMSE estimate, or other linear/non-linear estimate of ${{\mH}}_{j}^{(j)}$. For demonstration purpose only, in this paper we use the simplest LS estimate, which already shows very good performance.

In order to adapt the method in section \ref{subsec:singleUser} to multi-user scenario, we propose to first neutralize the intra-cell interference with a Zero-Forcing (ZF) filter ${\bf{T}}_{jk}$ based on the LS estimate $\widehat{{\mH}}_{j\backslash k}^{(j)}$, and then apply the spatial filter ${{\bf{\Xi}}_{jk}}$. After these two filters, the data signal is now
\begin{equation}\label{Eq:Wtilde2}
{\bf{\widetilde W}}_{jk} \triangleq {{\bf{\Xi}}_{jk}} {\bf{T}}_{jk} {\bf{W}}^{(j)},
\end{equation}
where
\begin{equation}\label{Eq:Tjk}
{\bf{T}}_{jk} \triangleq {\bf{I}}_M - \widehat{{\mH}}_{j\backslash k}^{(j)}(\widehat{{\mH}}_{j\backslash k}^{(j)H} \widehat{{\mH}}_{j\backslash k}^{(j)})^{-1}\widehat{{\mH}}_{j\backslash k}^{(j)H},
\end{equation}
and 
\begin{equation}\label{Eq:Xijk}
{{\bf{\Xi}}_{jk}} \triangleq {\left( {\sum\limits_{l = 1}^L {{{\bf{R}}_{lk}^{(j)}}}  + \sigma _n^2{{\bf{I}}_M}} \right)^{ - 1}} {{\bf{R}}_{jk}^{(j)}}.
\end{equation}
The rest of this method proceeds as in the single user setting.
Take the dominant eigenvector of ${\bf{\widetilde W}}_{jk} {\bf{\widetilde W}}_{jk}^{H} / C$:
\begin{equation}
\widetilde{\bf{u}}_{jk1} = \ve_1\{ \frac{1}{C} {\bf{\widetilde W}}_{jk} {\bf{\widetilde W}}_{jk}^{H} \}.
\end{equation}
The estimate of the direction of ${\vh}_{jk}^{(j)}$ is obtained by
\begin{equation}\label{Eq:u_bar_jk1}
\overline{\bf{u}}_{jk1} = \frac{{\bf{\Xi}}_{jk}' \widetilde{\bf{u}}_{jk1}}{\left\|{\bf{\Xi}}_{jk}' \widetilde{\bf{u}}_{jk1}\right\|_2},
\end{equation}
where
\begin{equation}
{\bf{\Xi}}_{jk}' \triangleq {{\bf{R}}_{jk}^{(j)\dag}}{\left( {\sum\limits_{l = 1}^L {{{\bf{R}}_{lk}^{(j)}}}  + \sigma _n^2{{\bf{I}}_M}} \right)}.
\end{equation}
Finally the phase and amplitude ambiguities are resolved by the training sequence, and we have the estimate of ${\vh}_{jk}^{(j)}$:
\begin{equation}\label{Eq:CAProjMU}
{\bf{\widehat {{h}}}}_{jk}^{(j)\text{CA}}= \frac{1}{\tau} \overline{\bf{u}}_{jk1} \overline{\bf{u}}^{H}_{jk1} \mathbf{Y}^{(j)} \mS^H.
\end{equation}
Note that in this method, we build the ZF type filter ${\bf{T}}_{jk}$ based on a rough LS estimate. Further improvements can be attained with higher quality estimates at the cost of
higher complexity. As a simple example, we can reduce the effect of noise on the estimate $\widehat{{\mH}}_{j\backslash k}^{(j)}$ by first applying EVD of $\mW^{(j)} \mW^{(j)H}/C$,
then removing the subspace where the noise lies, and finally performing LS estimation. These extensions are out of the scope of this paper.

\section{Low-complexity alternatives}\label{sec:lowComplexityAlter}

In this section, we propose two alternatives of the method shown in section \ref{sec:cova-aidedBlind}, aiming at lower computational complexity at the cost of mild performance losses.
\subsection{Subspace and amplitude based projection}
The low-rankness of channel covariance implies that the uplink received desired signal lives in a reduced subspace. By projecting the received data signal $\mW^{(j)}$ onto the signal space of ${\mathbf{R}}^{(j)}_{jk}$, we are able to preserve the signal from user $k$ in cell $j$ while removing the interference and noise that live in its complementary subspace. In the following, we show a subspace-based signal space projection method that relies on the covariance of desired channel only.
For ease of exposition, we simplify the system setup to single user per cell. Let the user in cell $j$ be the target user. The EVD of the covariance of the desired channel is
\begin{equation}
{\mathbf{R}}^{(j)}_{j} = \mV_{j} \ \bm{\Sigma}_{j} \mV_{j}^H,
\end{equation}
where the diagonal entries of $\bm{\Sigma}_{j}$ contains the non-negligible eigenvalues of ${\mathbf{R}}^{(j)}_{j}$.
Then we project the received data signal onto the signal space of ${\mathbf{R}}^{(j)}_{j}$, or the column space of $\mV_{j}$:
\begin{equation}
{\bf{\overline W}}_{j} \triangleq \mV_{j} \mV_{j}^H {\bf{W}}^{(j)}.
\end{equation}
The rest of this method follows the same idea as the covariance-aided amplitude based projection scheme. Taking the eigenvector corresponding to the largest eigenvalue of ${\bf{\overline W}}_{j}{{\bf{\overline W}}_{j}^{H}}/C$:
\begin{equation}
\overline{\bf{u}}_{j1} = \ve_1\left\{ \frac{1}{C} {\bf{\overline W}}_{j}{{\bf{\overline W}}_{j}^{H}} \right\},
\end{equation}
the channel estimate of ${{\mathbf{h}}}_{j}^{(j)}$ is given by
\begin{equation}\label{Eq:subspace+blind}
{{\widehat{{\mathbf{h}}}}_{j}^{(j)\text{SA}}} = \frac{1}{\tau} \overline{\bf{u}}_{j1} \overline{\bf{u}}^{H}_{j1} \mathbf{Y}^{(j)} \vs^*,
\end{equation}
where the superscript ``SA" stands for ``subspace and amplitude based projection". Note that this method does not require the covariance of interference channels or variance of noise.
It explicitly relies on the assumption that the desired covariance matrix has a low-dimensional signal subspace, with some degradations expected when this condition is not realized in practice.
In fact, if ${\mathbf{R}}^{(j)}_{j}$ has full rank, this method degrades to pure amplitude based projection.

Note that this ``SA" estimator has lower complexity than the ``CA" estimator (\ref{Eq:CAProj}) in the sense that 1) ``SA" estimator does not require the statistical knowledge of the interference channels or the variance of the noise, and 2) ``SA" estimator skips step 2 in Algorithm \ref{Alg:CA_estimator}.

The physical condition under which full decontamination is achieved with this method is shown below in the case of a ULA. We denote the angular support of desired channel ${\mathbf{h}}^{(j)}_{j}$ by $\Phi_d$ and the multipath components of the interference channel ${{\mathbf{h}}_{l}^{(j)}}$ falling in $\Phi_d$ as ${{\mathbf{h}}_{l\text{i}}^{(j)}}$.
\begin{theorem}\label{theoAnglePower2}
For a ULA base station, if the power of interference channel that falls into the angular support $\Phi_d$ satisfies
\begin{align}
\forall l \neq j, \left\| {{\mathbf{h}}_{l\text{i}}^{(j)}} \right\|_2 < \left\| {{\mathbf{h}}_{j}^{(j)}} \right\|_2,
\end{align}
and the channel covariance satisfies
\begin{align}\label{Eq:bound2}
\forall M \in \mathbb{Z}^+, \forall l \neq j, \left\| {\mathbf{R}}^{(j)\frac{1}{2}}_{j} \mV_{j} \mV_{j}^H {\mathbf{R}}^{(j)\frac{1}{2}}_{l} \right\|_2 < + \infty,
\end{align}
then, the estimation error of the estimator (\ref{Eq:subspace+blind}) vanishes
\begin{equation}\label{Eq:limit_partialOverlap2}
\mathop {\lim }\limits_{M, C \to \infty } \frac{\left\|{\bf{\widehat {{h}}}}_{j}^{(j)\text{SA}} - {{\vh}}_{j}^{(j)} \right\|_2^2}{\left\|{{\vh}}_{j}^{(j)}\right\|_2^2} = 0.
\end{equation}
\end{theorem}
\begin{proof}
\quad \emph{Proof:}
Due to lack of space, we skip the complete proof and only give two key steps below.
By applying the asymptotic orthogonality between two steering vectors which are associated with different AoAs ( Lemma 3 in \cite{yin:13}), we may readily obtain
\begin{align}
\mathop {\lim }\limits_{M \to \infty } \frac{1}{\sqrt{M}} \mV_{j} \mV_{j}^H {{\mathbf{h}}_{j}^{(j)}} &= \frac{1}{\sqrt{M}} {{\mathbf{h}}_{j}^{(j)}} \\
\mathop {\lim }\limits_{M \to \infty } \frac{1}{\sqrt{M}} \mV_{j} \mV_{j}^H {{\mathbf{h}}_{l}^{(j)}} &= \frac{1}{\sqrt{M}} {{\mathbf{h}}_{l\text{i}}^{(j)}}, \forall l \neq j,
\end{align}
which means the multipath components of interference that fall outside $\Phi_d$ disappear asymptotically after the projection by $\mV_{j} \mV_{j}^H$.
Then, equation (\ref{Eq:bound2}) ensures that
\begin{align}\label{Eq:hjhl=0}
\mathop {\lim }\limits_{M \to \infty } \frac{1}{M} {{\underline{\mathbf{h}}}_{j}^{(j)H}} {{\underline{\mathbf{h}}}_{l}^{(j)}} = 0, l \neq j,
\end{align}
where
\begin{align}
{{\underline{\mathbf{h}}}_{l}^{(j)}} \triangleq \mV_{j} \mV_{j}^H {{\mathbf{h}}_{l}^{(j)}}, l = 1, \ldots, L.
\end{align}

\end{proof}
Note that in Theorem \ref{theoAnglePower2} condition (\ref{Eq:bound2}) is less restrictive than the uniformly boundedness of the spectral norm of the channel covariance. 
In the special case of zero angular spread, the rank of channel covariance becomes one. Denote the deterministic AoA from the user in cell $l$ to base station $j$ as $\overline{\theta}_{l}^{(j)}$. We can easily see that the channel estimation error of (\ref{Eq:subspace+blind}) vanishes completely as $M, C \rightarrow \infty$ as long as
\begin{equation}
\forall l \neq j, \overline{\theta}_{l}^{(j)} \neq \overline{\theta}_{j}^{(j)},
\end{equation}
which occurs with probability one.

When channel covariance is not available, we can still benefit from the subspace projection method by approximating $\mV_{j}$ with a subset of discrete Fourier transform (DFT) basis as shown in \cite{yin:15a}. This DFT basis can be chosen based on a small number of channel observations. The generalization to multi-user case can be done by introducing the ZF filter (\ref{Eq:Tjk}) as in section \ref{subsec:multiuser}. Due to lack of space, we skip the details.

\subsection{MMSE + amplitude based projection}
Another alternative is to directly project the MMSE estimate onto the subspace of $\mE^{(j)}$ obtained by EVD of $\mW^{(j)} \mW^{(j)H}/C$ as in section \ref{sec:blind}.
The estimator for the multi-user channel $\mH_j^{(j)}$ is given by
\begin{equation}\label{Eq:MMSE+Blind}
\widehat{{{\mH}}}_j^{(j){\text{MA}}} = \overline{\mE}^{(j)} \overline{\mE}^{(j)H} \overline{\mR}^{(j)}_j  \left( \tau (\sum_{l=1}^{L} \overline{\mR}^{(j)}_l) + \sigma_n^2 \mI_{K M} \right)^{-1} \overline{\mS}^H \vy^{(j)},
\end{equation}
where
\begin{align}
\overline{\mE}^{(j)} \triangleq  \mI_K  \otimes {\mE}^{(j)},
\end{align}
\begin{equation}
\begin{array}{*{20}{l}}
  {\mathbf{\overline S}} &\triangleq \mS^T \otimes \mI_M = \left[ {\begin{array}{*{20}{c}}
  {{{\mathbf{s}}_1} \otimes {{\mathbf{I}}_M}}& \cdots &{{{\mathbf{s}}_K} \otimes {{\mathbf{I}}_M}}
\end{array}} \right],
\end{array}
\end{equation}
and
\begin{equation}
{\overline{\mathbf{R}}_{l}^{(j)}} = {\mathop{\rm diag}\nolimits} \{ {{\mathbf{R}}^{(j)}_{l1},...,{\mathbf{R}}^{(j)}_{lK}}\}, l = 1, \ldots, L.
\end{equation}
The superscript ``MA" denotes MMSE + amplitude based projection.
It is worth noting that both the amplitude-based projection and angular-based projection require a large number of antennas to achieve complete decontamination. In contrast, the MMSE estimator is efficient with very small number of antennas. As $M$ grows, MMSE estimator starts to reduce interference earlier than the previously proposed methods, as will be shown by simulations in Section \ref{sec:numericalResult}.
However, unlike the previously proposed schemes, this ``MA" estimator cannot achieve complete decontamination when the interference channel is overlapping with desired channel in both angular and power domains.

\section{Numerical Results}\label{sec:numericalResult}
This section contains numerical results of our different channel estimation schemes compared with prior methods. In the simulation, we have multiple hexagonally shaped adjacent cells in the network. The radius of each cell is 1000 meters. Each base station has $M$ antennas, which forms a ULA, with half wavelength antenna spacing. The length of pilot sequence is $\tau = 10$.

Two performance metrics are considered. The first is the normalized channel estimation error
\begin{equation}\label{Eq:err}
\epsilon \triangleq  {\frac{1}{KL}}\sum\limits_{j = 1}^L \sum\limits_{k = 1}^K \left({\frac{\left\| {{{\widehat{{\mathbf{h}}}}_{jk}^{(j)}} - {{\mathbf{h}}}_{jk}^{(j)}} \right\|_2^2} {\left\| {{\mathbf{h}}}_{jk}^{(j)} \right\|_2^2 }}\right).
\end{equation}
The estimation errors in the plots are obtained by Monte Carlo simulations and displayed in dB scale.

The second metric is the uplink per-cell rate when MRC receiver (based on the obtained channel estimate) is used at the base station side.

In all simulations presented in this section, we assume that the channel covariance matrix is estimated using 1000 exact channel realizations.
The multipath angle of arrival of any channel (including the interference channel) follows a uniform distribution centered at the direction corresponding to line-of-sight (LoS). The number of multipath is $P=50$.
According to the coherence time model in \cite{rappaport1996}, for a mobile user moving at a vehicular speed of 70 km/h in an environment of 2.6 GHz carrier frequency and $5\mu$s high delay spread (corresponding to an excess distance of 1.5 km), the channel can be assumed coherent over 500 transmitted symbols. Thus, we will let $C=500$ in simulations, although larger coherence time can be expected in practice for a user with lower mobility.

Note that in all simulations, the amplitude-based projection and MMSE + amplitude based projection follow the enhanced eigenvector selection strategy shown in Remark \ref{remark1}
with the design parameter $\mu = 0.2$.

We first illustrate Theorem \ref{theoStatis} in Fig. \ref{fig:msePartialOverlapping}. Suppose we have a two-cell network, with each cell having one user. In order to make the interference overlapping in power domain with the desired signal, we set the path loss exponent $\gamma = 0$. The power of the interference channel has equal probability to be higher or lower than the power of the desired channel. The user in each cell is deliberately put in a symmetrical position such that the multipath angular supports of the interference and the desired channel are half overlapping with each other.

\begin{figure}[h]
  \centering
  \includegraphics[width=3.2in]{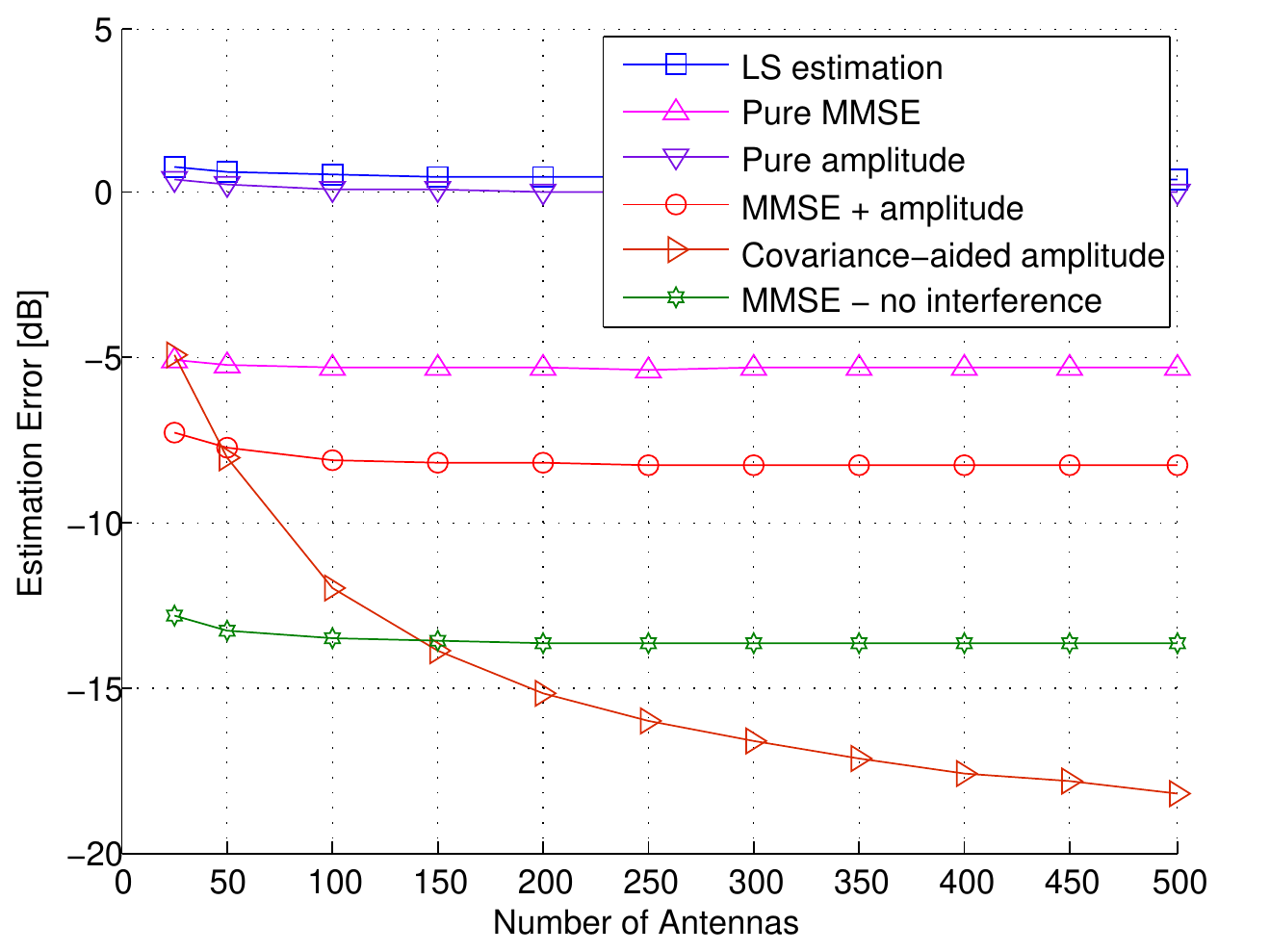}\\
  \caption{Estimation performance vs. M, 2-cell network, 1 user per cell, path loss exponent $\gamma = 0$, partially overlapping angular support, AoA spread 60 degrees, SNR = 0 dB.} \label{fig:msePartialOverlapping}
\end{figure}

In the figure, ``LS estimation" and ``Pure MMSE" denote the system performances when an LS estimator and an MMSE estimator (\ref{Eq:MMSE_2}) are used respectively. ``Pure amplitude" denotes the case when we apply the generalized amplitude based projection method only. ``MMSE + amplitude" represents the proposed estimator (\ref{Eq:MMSE+Blind}). ``Covariance-aided amplitude" denotes the proposed covariance-aided amplitude based projection method (\ref{Eq:CAProj}).
The curve ``MMSE - no interference" shows the estimation error of an MMSE estimator in an interference-free scenario. As can be seen from Fig. \ref{fig:msePartialOverlapping}, due to the overlapping interference in both angle and power domains, the performance of all estimators saturate quickly with the number of antennas, except the proposed covariance-aided amplitude based projection method, which eventually outperforms interference-free MMSE estimation.\footnote{The reason is that the performance of the interference-free MMSE estimation has a non-vanishing lower bound due to white Gaussian noise. On the contrary, our proposed covariance-aided amplitude based projection method eliminates the effects of noise and interference asymptotically.}

\begin{figure}[h]
  \centering
  \includegraphics[width=3.2in]{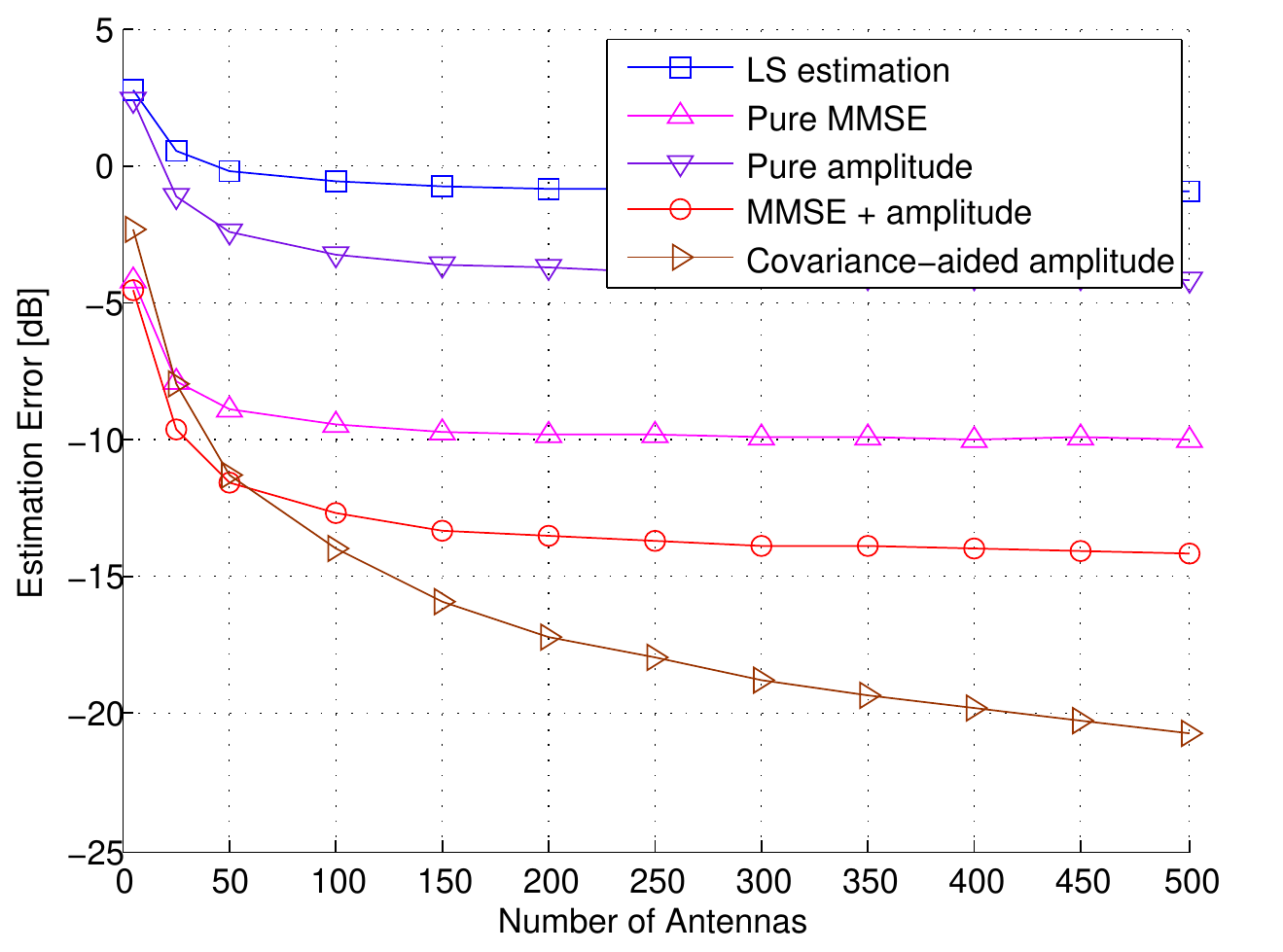}\\
  \caption{Estimation performance vs. M, 7-cell network, one user per cell, AoA spread 30 degrees, path loss exponent $\gamma = 2$, cell-edge SNR = 0 dB.} \label{fig:MSE.vs.M_7cell_1user}
\end{figure}

\begin{figure}[h]
  \centering
  \includegraphics[width=3.2in]{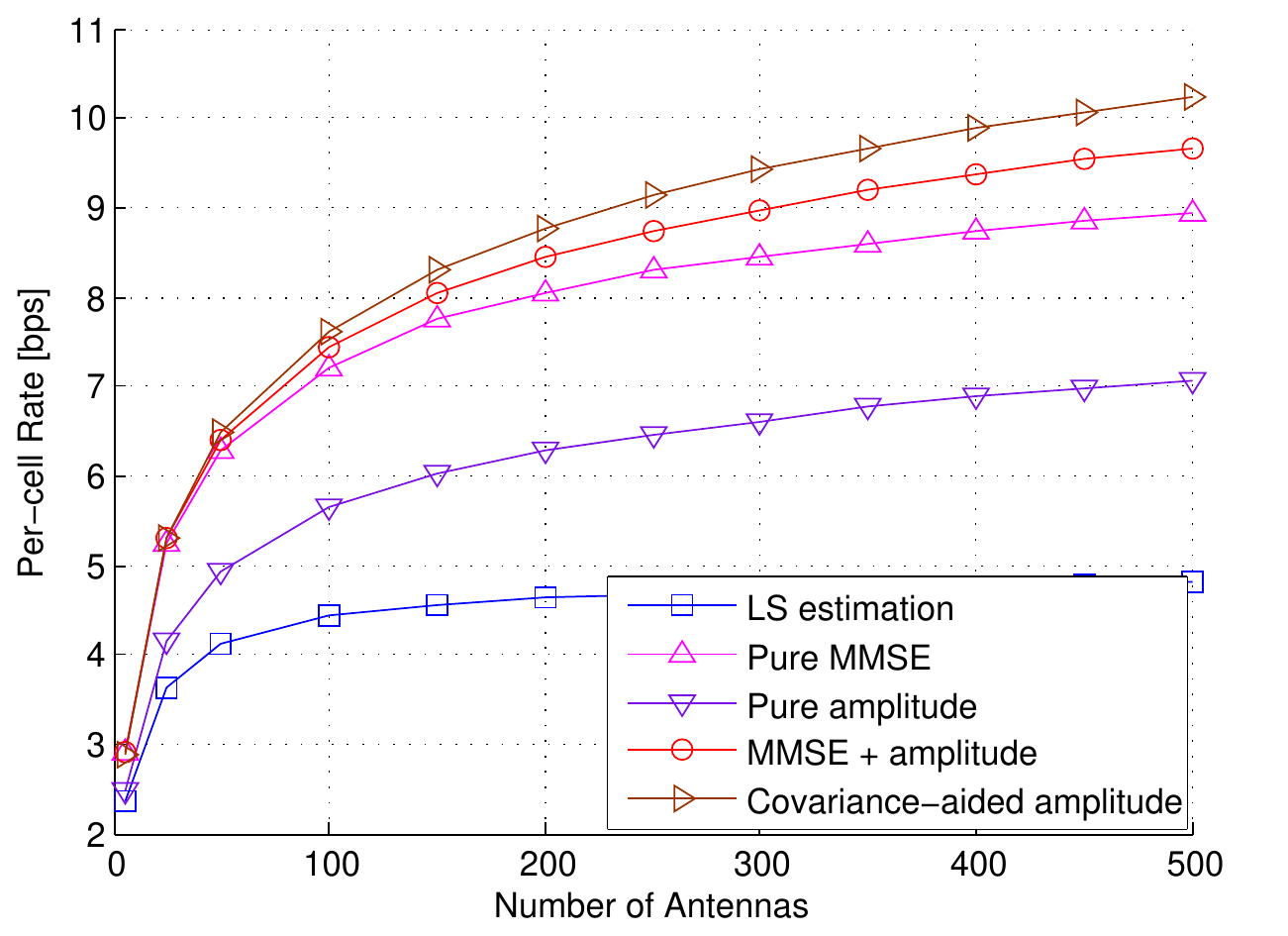}\\
  \caption{Uplink per-cell rate vs. M, 7-cell network, one user per cell, AoA spread 30 degrees, path loss exponent $\gamma = 2$, cell-edge SNR = 0 dB.} \label{fig:sumrate.vs.M_7cell_1user}
\end{figure}

In Fig. \ref{fig:MSE.vs.M_7cell_1user} and Fig. \ref{fig:sumrate.vs.M_7cell_1user}, we show the performance of estimation error and the corresponding uplink per-cell rate for a 7-cell network, with single user per cell. The users are assumed to be distributed randomly and uniformly within their own cells excluding a central disc with radius 100 meters. The angular spread of the user channel (including interference channel) is 30 degrees. The path loss exponent is now $\gamma = 2$. As we may observe, the traditional LS estimator suffers from severe pilot contamination.
The pure amplitude based method and the pure MMSE method alleviate the pilot interference, yet saturate with the number of antennas. These saturation effects come from the overlapping of the interference and the desired channels in power and angular domains respectively. The ``MMSE + amplitude" approach outperforms these two known methods as it discriminates against interference in both amplitude and angular domains. However this scheme cannot cope with the case of overlapping in both domains.
Owing to its robustness, the covariance-aided amplitude projection method outperforms the rest in terms of both estimation error and uplink per-cell rate.

We now turn our attention to multi-cell multi-user scenario. Fig. \ref{fig:MSE.vs.M_7cell_4user} and Fig. \ref{fig:sumrate.vs.M_7cell_4user} show the channel estimation performance and the corresponding uplink per-cell rate for a 7-cell network with each cell having 4 users. In these two figures, we add the curve of subspace and amplitude based projection, which is denoted in the figures as ``Subspace + amplitude".
The other parameters remain unchanged compared with those in Fig. \ref{fig:MSE.vs.M_7cell_1user} and Fig. \ref{fig:sumrate.vs.M_7cell_1user}. We can notice that in Fig. \ref{fig:MSE.vs.M_7cell_4user} the covariance-aided amplitude projection method has some performance loss with respect to the low-complexity MMSE + amplitude method and the MMSE method when the number of antennas is small. It is due to the following two facts: 1) when $M$ is small, it is well known that MMSE works well, but not the amplitude based methods, and 2) with small $M$, the asymptotical orthogonality of channels of different users is not fully exhibited, and consequently a small amount of signal of interest is removed by the ZF filter $\mT_{kj}$, along with intra-cell interference.
However it is not disturbing in the sense that 1) as the number of antennas grows, the covariance-aided amplitude projection method quickly outperforms the other methods; and 2) The per-cell rate of this proposed method is still good even with moderate number of antennas, e.g., $M>25$.
It is also interesting to note that the low-complexity alternative scheme, subspace and amplitude based projection method, has some minor performance loss, yet keeps approximately the same slope as the covariance-aided amplitude projection.

\begin{figure}[h]
  \centering
  \includegraphics[width=3.2in]{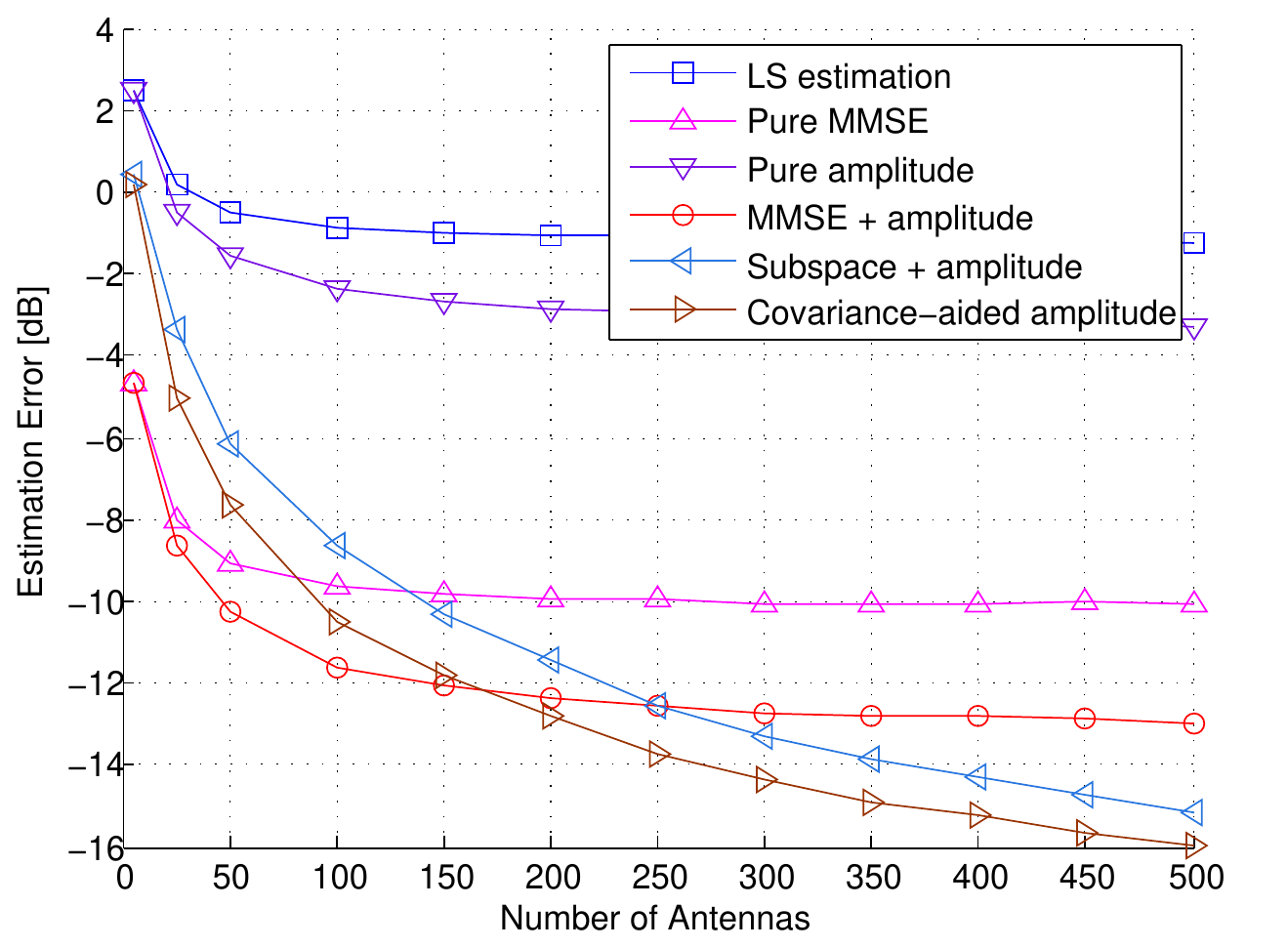}\\
  \caption{Estimation performance vs. M, 7-cell network, 4 users per cell, AoA spread 30 degrees, path loss exponent $\gamma = 2$, cell-edge SNR = 0 dB.} \label{fig:MSE.vs.M_7cell_4user}
\end{figure}

\begin{figure}[h]
  \centering
  \includegraphics[width=3.2in]{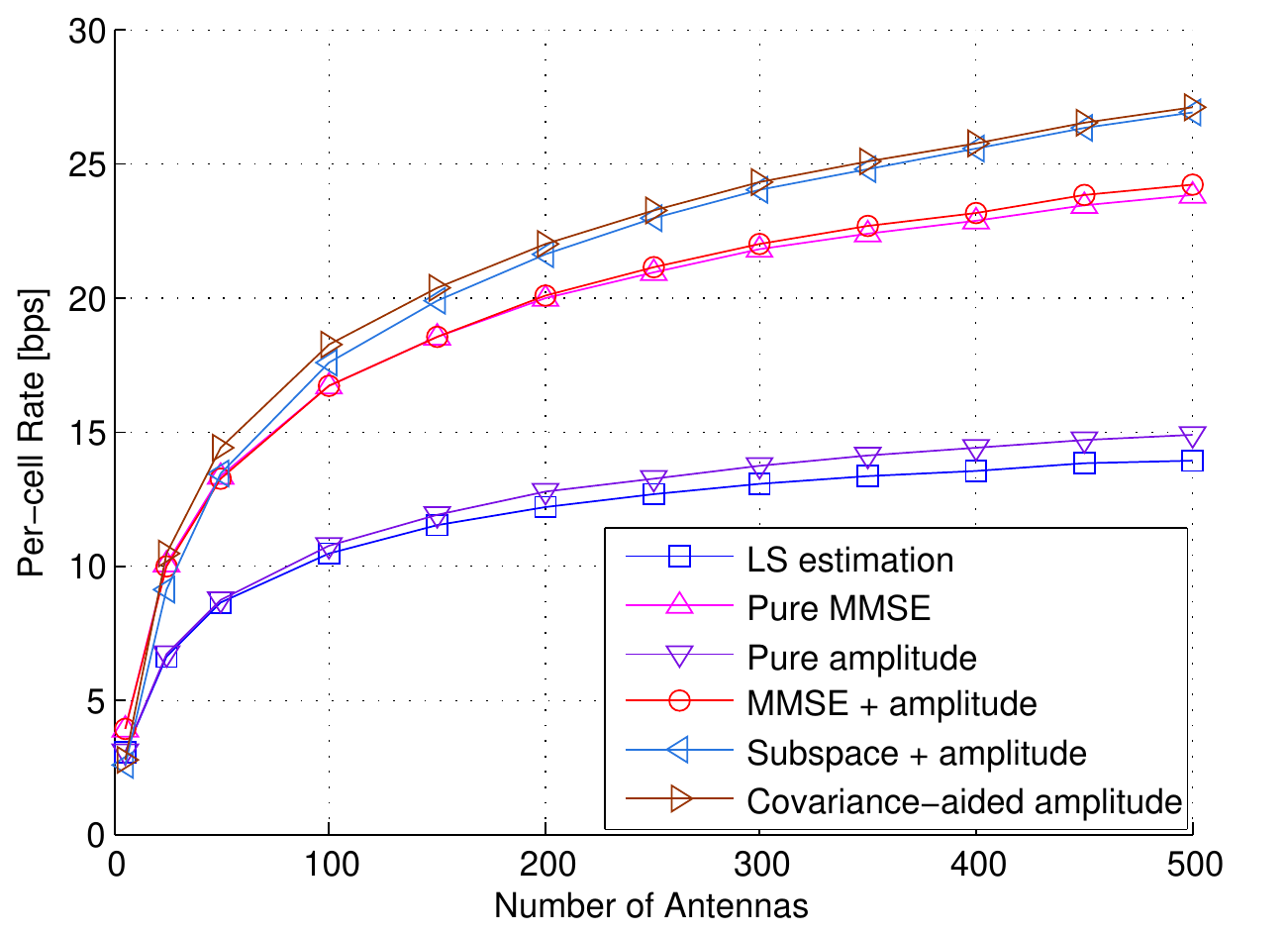}\\
  \caption{Uplink per-cell rate vs. M, 7-cell network, 4 users per cell, AoA spread 30 degrees, path loss exponent $\gamma = 2$, cell-edge SNR = 0 dB.} \label{fig:sumrate.vs.M_7cell_4user}
\end{figure}

\section{Conclusions}\label{conclusion}
In this paper we proposed a series of robust channel estimation algorithms exploiting path diversity in both angle and amplitude domains.
The first method called ``covariance-aided amplitude based projection" is robust even when the desired channel and the interference channels overlap in multipath AoA and are not separable just in terms of power. Two low-complexity alternative schemes were proposed, namely ``subspace and amplitude based projection" and ``MMSE + amplitude based projection". Asymptotic analysis shows the condition under which the channel estimation error converges to zero.

\appendix
\subsection{Proof of Proposition \ref{propBoundedness}:}\label{proof:propBoundedness}
\begin{proof}
Denote the associated path loss as ${\beta}$. The covariance ${\mathbf{R}}$ is a Toeplitz matrix, with its $mn$-th entry given by
\begin{align}\label{Eq:CorrUni}
{{\mathbf{R}}}(m,n) &= {\beta}\int_{0}^{{\pi }} p(\theta) {e^{j2\pi \frac{D}{\lambda }(n - m)\cos (\theta )}} d\theta \\
& = {\beta}\int_{-1}^{{1}} p\left(\arccos(x)\right) {e^{j2\pi \frac{D}{\lambda }(n - m)x}} \frac{1}{\sqrt{1-x^2}} d x \nonumber \\
& = \frac{\beta \lambda}{2 \pi D} \int_{-{\frac{2 \pi D}{\lambda}}}^{{\frac{2 \pi D}{\lambda}}} \frac{p\left(\arccos(\frac{\lambda x }{2 \pi D})\right)}{\sqrt{1- \left( \frac{\lambda x }{2 \pi D}\right)^2}} {e^{j (n - m)x}} d x, \nonumber \\
& = \frac{1}{2 \pi} \int_{-{\frac{2 \pi D}{\lambda}}}^{{\frac{2 \pi D}{\lambda}}} f(x) {e^{j (n - m)x}} d x,
\end{align}
where
\begin{align}
f(x) \triangleq \frac{\beta \lambda}{D} \frac{p\left(\arccos(\frac{\lambda x }{2 \pi D})\right)}{\sqrt{1- \left( \frac{\lambda x }{2 \pi D}\right)^2}}.
\end{align}
Since $0, \pi \notin \Phi$, or in other words, $p(0) = p(\pi) = 0$, and that $p(\theta) < \infty, \forall \theta \in \Phi$, it follows that $f(x)$ is uniformly bounded:
\begin{equation}
f(x) < + \infty,  -{\frac{2 \pi D}{\lambda}} \leq x \leq {\frac{2 \pi D}{\lambda}}.
\end{equation}
Thus, the Toeplitz matrix ${\mathbf{R}}$ is related to the real integrable and uniformly bounded generating function $f(x)$, with its entries being Fourier coefficients of $f(x)$. We now resort to the known result on the spectrum of the $n \times n$ Toeplitz matrices $\mT_n(f)$ defined by the generating function $f(x)$. Denote by $\text{ess inf}$ and $\text{ess sup}$ the essential minimum and the essential maximum of $f$, i.e., the infimum and the supremum of $f$ up to within a set of measure zero. Let $m_f \triangleq \text{ess inf} f$ and $M_f \triangleq \text{ess sup} f$.
\begin{theorem}\label{theoToepSpectrum}
\cite{grenander1985} If $\lambda_0^{(n)} \leq \lambda_1^{(n)} \leq \cdots \leq \lambda_{n-1}^{(n-1)}$ are the eigenvalues of $\mT_n(f)$, then, the spectrum of $\mT_n(f)$ is contained in $(m_f, M_f)$; moreover
$\mathop {\lim }\limits_{n \to \infty } \lambda_0^{(n)} = m_f$ and $\mathop {\lim }\limits_{n \to \infty } \lambda_{n-1}^{(n-1)} = M_f$.
\end{theorem}
By invoking Theorem \ref{theoToepSpectrum}, we obtain that $\mathop {\lim }\limits_{M \to \infty } \left\|{{\bf{R}}}\right\|_2 = M_f < \infty$. In addition, for any finite $M$, the inequality $\left\|{{\bf{R}}}\right\|_2 < \infty$ always holds true. This concludes the proof.
\end{proof}

\subsection{Proof of Lemma \ref{lemmaSpectralNormPerturbation}:}\label{proof:lemmaSpectralNormPerturbation}
\begin{proof}
Since $\mR_j$ and ${\left( \sum_{l = 1}^L {\bf{R}}_{l} + \sigma _n^2{{\bf{I}}_M} \right)^{ - 1}}$ are both positive semi-definite (PSD) Hermitian matrices, we can directly apply the inequalities of \cite{marshall2010} on the eigenvalues of the product of two PSD Hermitian matrices
\begin{align}
\left\| {\bf{\Xi}}_j \right\|_2 &\leq  \left\| {\left( \sum\limits_{l = 1}^L {\bf{R}}_{l} + \sigma _n^2{{\bf{I}}_M} \right)^{ - 1}} \right\|_2 \left\| \mR_j \right\|_2
< \frac{\zeta }{\sigma _n^2}.
\end{align}
It is straightforward to show that
\begin{equation}\label{Eq:boundXiXi}
\left\| {\bf{\Xi}}_j{\bf{\Xi}}_j^H \right\|_2 = \left\| {\bf{\Xi}}_j \right\|_2^2 < \frac{\zeta^2}{\sigma _n^4},
\end{equation}
which indicates that the spectral norm of ${\bf{\Xi}}_j{\bf{\Xi}}_j^H$ is also uniformly bounded. This proves Lemma \ref{lemmaSpectralNormPerturbation}.
\end{proof}

\subsection{Proof of Lemma \ref{lemmaInnerProd}:}\label{proof:lemmaInnerProd}
\begin{proof}
Using the spatial correlation model (\ref{Eq:spaWhi}), we may write
\begin{align}
\frac{1}{M} {\underline{\vh}_j^H} \underline{\vh}_l &= \frac{1}{M} {{\mathbf{h}}^H_{\text{W}}}_{j} {\mathbf{R}}_{j}^{\frac{1}{2}} {\bf{\Xi}}_j^H {\bf{\Xi}}_j {{\mR}}_{l}^{\frac{1}{2}} {{\vh}_{\text{W}}}_{l}.
\end{align}
By an abuse of notation, we now use the operator $\lambda_1\{ \cdot \}$ to represent the largest singular value of a matrix.
Appealing to the singular value inequalities in \cite{roger:94}, we can show that the maximum singular value of ${\mathbf{R}}_{j}^{\frac{1}{2}} {\bf{\Xi}}_j^H {\bf{\Xi}}_j {{\mR}}_{l}^{\frac{1}{2}} $ yields
\begin{align}
\lambda_1\{ {\mathbf{R}}_{j}^{\frac{1}{2}} {\bf{\Xi}}_j^H {\bf{\Xi}}_j {{\mR}}_{l}^{\frac{1}{2}}  \} &\leq  \lambda_1\{ {\mathbf{R}}_{j}^{\frac{1}{2}} \} \lambda_1\{ {\bf{\Xi}}_j^H {\bf{\Xi}}_j {{\mR}}_{l}^{\frac{1}{2}}\} \\
& < \zeta^{\frac{1}{2}} \lambda_1\{ {\bf{\Xi}}_j^H {\bf{\Xi}}_j \} \lambda_1\{{{\mR}}_{l}^{\frac{1}{2}}\} \\
&< \frac{\zeta^3}{\sigma _n^4},
\end{align}
which means the spectral radius of the complex matrix ${\mathbf{R}}_{j}^{\frac{1}{2}} {\bf{\Xi}}_j^H {\bf{\Xi}}_j {{\mR}}_{l}^{\frac{1}{2}}$ is uniformly bounded for any $M$. Thus, according to Lemma \ref{lemmaQuadratic}, $\frac{1}{M} {\underline{\vh}_j^H} \underline{\vh}_l, \forall l \neq j,$ converges almost surely to zero. Thus (\ref{Eq:hjhl_underline}) holds true. In a similar way, we can prove (\ref{Eq:hlhl_underline}).
This concludes the proof of Lemma \ref{lemmaInnerProd}.
\end{proof}

\subsection{Proof of Lemma \ref{lemmaEigenvector}:}\label{proof:lemmaEigenvector}
\begin{proof}
Define
\begin{align}\label{Eq:Gamma}
\Gamma &\triangleq \mathop {\lim }\limits_{C \to \infty } \left({\frac{1}{C} { {\bf{\widetilde W}}_j{{\bf{\widetilde W}}_j^{H}} }} \right) \\
&=\underline{\vh}_j \underline{\vh}_j^H + \sum_{l \neq j}\underline{\vh}_l \underline{\vh}_l^H + \sigma_n^2 {\bf{\Xi}}_j{\bf{\Xi}}_j^H.
\end{align}
In this proof, we first consider the noise free scenario and let
\begin{equation}
\Gamma_\text{nf} = \underline{\vh}_j \underline{\vh}_j^H + \sum_{l \neq j}\underline{\vh}_l \underline{\vh}_l^H,
\end{equation}
where the subscript ``nf" denotes noise free.
We can then write
\begin{align}\label{Eq:noNoise}
& \quad  \mathop {\lim }\limits_{M \to \infty } \left\| \frac{\Gamma_\text{nf}}{M} \frac{\underline{\vh}_j}{\left\|\underline{\vh}_j\right\|_2} - \alpha_j \frac{\underline{\vh}_j}{\left\|\underline{\vh}_j\right\|_2} \right\|_2^2 \\
& = \mathop {\lim }\limits_{M \to \infty } ( \frac{\Gamma_\text{nf}}{M} \frac{\underline{\vh}_j}{\left\|\underline{\vh}_j\right\|_2} - \alpha_j \frac{\underline{\vh}_j}{\left\|\underline{\vh}_j\right\|_2} )^H ( \frac{\Gamma_\text{nf}}{M} \frac{\underline{\vh}_j}{\left\|\underline{\vh}_j\right\|_2} - \alpha_j \frac{\underline{\vh}_j}{\left\|\underline{\vh}_j\right\|_2} ) \nonumber \\
& = \mathop {\lim }\limits_{M \to \infty }  \frac{1}{M^2} \left\|\underline{\vh}_j\right\|_2^2 - \mathop {\lim }\limits_{M \to \infty }  \frac{2\alpha_j}{M} \left\|\underline{\vh}_j\right\|_2^2 + \alpha_j^2 \nonumber \\
& = \alpha_j^2 - 2\alpha_j^2 + \alpha_j^2 \nonumber \\
& = 0, \nonumber
\end{align}
which proves that when $M \rightarrow \infty$, an eigenvalue of the random matrix $\Gamma_\text{nf}/M$ converges to $\alpha_j$, with its corresponding eigenvector converging to ${\underline{\vh}_j}/{\left\|\underline{\vh}_j\right\|_2}$ up to a random phase.

Then we consider the Hermitian matrix $\sigma_n^2 {\bf{\Xi}}_j{\bf{\Xi}}_j^H$ as a perturbation on $\Gamma_\text{nf}/M$. 
Due to the Bauer-Fike Theorem \cite{bauer1960} on the perturbation of eigenvalues of Hermitian matrices, together with Lemma \ref{lemmaSpectralNormPerturbation}, we have for $1 \leq i \leq L$:
\begin{align}
&\mathop {\lim }\limits_{M \to \infty } \left| \lambda_i\left\{ \frac{\Gamma}{M} \right\} - \lambda_i\left\{ \frac{\Gamma_\text{nf}}{M} \right\}\right| \\
&\leq \mathop {\lim }\limits_{M \to \infty } \frac{\sigma_n^2}{M} \left\| {\bf{\Xi}}_j{\bf{\Xi}}_j^H \right\|_2  \\
& = 0.
\end{align}
The above result shows that the impact of the perturbation on the eigenvalues of $\Gamma_\text{nf}/M$ vanishes as $M \rightarrow \infty$. In other words, $\alpha_j$ is again an asymptotic eigenvalue of ${\Gamma}/{M}$. Now we verify that despite the perturbation, the eigenvector of ${\Gamma}/{M}$ corresponding to the asymptotic eigenvalue $\alpha_j$ also converges to ${\underline{\vh}_j}/{\left\|\underline{\vh}_j\right\|_2}$ up to a random phase. To prove this, it is sufficient to show that
\begin{align}
& \quad  \mathop {\lim }\limits_{M \to \infty } \left\| \frac{\Gamma}{M} \frac{\underline{\vh}_j}{\left\|\underline{\vh}_j\right\|_2} - \alpha_j \frac{\underline{\vh}_j}{\left\|\underline{\vh}_j\right\|_2} \right\|_2 \\
& \leq \mathop {\lim }\limits_{M \to \infty } \left\| \frac{\Gamma_\text{nf}}{M} \frac{\underline{\vh}_j}{\left\|\underline{\vh}_j\right\|_2} - \alpha_j \frac{\underline{\vh}_j}{\left\|\underline{\vh}_j\right\|_2} \right\|_2 + \left\| \frac{\sigma_n^2 {\bf{\Xi}}_j{\bf{\Xi}}_j^H}{M} \frac{\underline{\vh}_j}{\left\|\underline{\vh}_j\right\|_2}\right\|_2 \nonumber \\
& \overset{{(a)}}{=} 0, \nonumber
\end{align}
where $(a)$ is due to the definition of the spectral norm
\begin{align}
\mathop {\lim }\limits_{M \to \infty } \left\| \frac{\sigma_n^2 {\bf{\Xi}}_j{\bf{\Xi}}_j^H}{M} \frac{\underline{\vh}_j}{\left\|\underline{\vh}_j\right\|_2}\right\|_2 = 0.
\end{align}
It follows that
\begin{align}
\mathop {\lim }\limits_{M,C \to \infty } \left\| \frac{ {\bf{\widetilde W}}_j{{\bf{\widetilde W}}_j^{H}} }{MC} \frac{\underline{\vh}_j}{\left\|\underline{\vh}_j\right\|_2} - \alpha_j \frac{\underline{\vh}_j}{\left\|\underline{\vh}_j\right\|_2} \right\|_2 &= 0,
\end{align}
which concludes the proof of Lemma \ref{lemmaEigenvector}.
\end{proof}

\subsection{Proof of Lemma \ref{lemmaEigenvectorOrig}:}\label{proof:lemmaEigenvectorOrig}
\begin{proof}
We can derive
\begin{align}\label{Eq:h_udl_converge2}
& \quad \mathop {\lim }\limits_{M,C \to \infty } {\left\| \frac{{\bf{\Xi}}_j' \underline{\vh}_j}{\left\| {\bf{\Xi}}_j' \underline{\vh}_j \right\|_2} - \frac{{\bf{\Xi}}_j' \widetilde{\vu}_{j1} e^{j\phi}}{\left\| {\bf{\Xi}}_j' \widetilde{\vu}_{j1} \right\|_2}  \right\|_2^2 } \\
& = \mathop {\lim }\limits_{M,C \to \infty } \left( \frac{{\bf{\Xi}}_j' \underline{\vh}_j}{\left\| {\bf{\Xi}}_j' \underline{\vh}_j \right\|_2} - \frac{{\bf{\Xi}}_j' \widetilde{\vu}_{j1} e^{j\phi}}{\left\| {\bf{\Xi}}_j' \widetilde{\vu}_{j1} \right\|_2}  \right)^H  \cdot \nonumber \\
& \qquad  \qquad   \left( \frac{{\bf{\Xi}}_j' \underline{\vh}_j}{\left\| {\bf{\Xi}}_j' \underline{\vh}_j \right\|_2} - \frac{{\bf{\Xi}}_j' \widetilde{\vu}_{j1} e^{j\phi}}{\left\| {\bf{\Xi}}_j' \widetilde{\vu}_{j1} \right\|_2}  \right) \nonumber \\
& = 2 - \mathop {\lim }\limits_{M,C \to \infty } \left( \frac{\underline{\vh}_j^H {{\bf{\Xi}}_j'}^H {\bf{\Xi}}_j' \widetilde{\vu}_{j1} e^{j\phi}}{\left\| {\bf{\Xi}}_j' \underline{\vh}_j \right\|_2 {\left\| {\bf{\Xi}}_j' \widetilde{\vu}_{j1} \right\|_2}} \right. \nonumber \\
& \quad \left. + \frac{e^{-j\phi}\widetilde{\vu}_{j1}^H {{\bf{\Xi}}_j'}^H {{\bf{\Xi}}_j'} \underline{\vh}_j}{{\left\| {\bf{\Xi}}_j' \underline{\vh}_j \right\|_2 {\left\| {\bf{\Xi}}_j' \widetilde{\vu}_{j1} \right\|_2}}}\right) \nonumber
\end{align}
We treat the following quantity separately
\begin{align}
& \quad \mathop {\lim }\limits_{M,C \to \infty } \frac{\underline{\vh}_j^H {{\bf{\Xi}}_j'}^H {\bf{\Xi}}_j' \widetilde{\vu}_{j1} e^{j\phi}}{\left\| {\bf{\Xi}}_j' \underline{\vh}_j \right\|_2 {\left\| {\bf{\Xi}}_j' \widetilde{\vu}_{j1} \right\|_2}} \\
& = \mathop {\lim }\limits_{M,C \to \infty } \frac{\underline{\vh}_j^H {{\bf{\Xi}}_j'}^H {\bf{\Xi}}_j' \left( \frac{\underline{\vh}_j}{\left\|\underline{\vh}_j\right\|_2}  + \widetilde{\vu}_{j1} e^{j\phi} - \frac{\underline{\vh}_j}{\left\|\underline{\vh}_j\right\|_2}\right)}{\left\| {\bf{\Xi}}_j' \underline{\vh}_j \right\|_2 {\left\| {\bf{\Xi}}_j' \widetilde{\vu}_{j1} \right\|_2}} \nonumber \\
& = \mathop {\lim }\limits_{M,C \to \infty } \frac{\left\| {\bf{\Xi}}_j' \frac{\underline{\vh}_j}{\left\| \underline{\vh}_j \right\|_2} \right\|_2}{{\left\| {\bf{\Xi}}_j' \widetilde{\vu}_{j1} \right\|_2}} \nonumber 
\end{align}
\begin{align}
& = \mathop {\lim }\limits_{M,C \to \infty } \frac{\left\| {\bf{\Xi}}_j' \left( \frac{\underline{\vh}_j}{\left\| \underline{\vh}_j \right\|_2} - \widetilde{\vu}_{j1} e^{j\phi} + \widetilde{\vu}_{j1} e^{j\phi} \right)\right\|_2}{{\left\| {\bf{\Xi}}_j' \widetilde{\vu}_{j1} \right\|_2}} \nonumber \\
& \leq \mathop {\lim }\limits_{M,C \to \infty } \frac{\left\| {\bf{\Xi}}_j' ( \frac{\underline{\vh}_j}{\left\| \underline{\vh}_j \right\|_2} - \widetilde{\vu}_{j1} e^{j\phi} )\right\|_2}{{\left\| {\bf{\Xi}}_j' \widetilde{\vu}_{j1} \right\|_2}} + \mathop {\lim }\limits_{M,C \to \infty } \frac{\left\| {\bf{\Xi}}_j' \widetilde{\vu}_{j1} e^{j\phi} \right\|_2}{{\left\| {\bf{\Xi}}_j' \widetilde{\vu}_{j1} \right\|_2}} \nonumber \\
& = 1 \label{Eq:hdivu}
\end{align}
In a similar way, we can prove that
\begin{align}\label{Eq:uhinv}
\mathop {\lim }\limits_{M,C \to \infty } \frac{{\left\| {\bf{\Xi}}_j' \widetilde{\vu}_{j1} \right\|_2}}{\left\| {\bf{\Xi}}_j' \frac{\underline{\vh}_j}{\left\| \underline{\vh}_j \right\|_2} \right\|_2} \leq 1.
\end{align}
Combining (\ref{Eq:hdivu}) and (\ref{Eq:uhinv}), we obtain
\begin{equation}\label{Eq:hXiu}
\mathop {\lim }\limits_{M,C \to \infty } \frac{\underline{\vh}_j^H {{\bf{\Xi}}_j'}^H {\bf{\Xi}}_j' \widetilde{\vu}_{j1} e^{j\phi}}{\left\| {\bf{\Xi}}_j' \underline{\vh}_j \right\|_2 {\left\| {\bf{\Xi}}_j' \widetilde{\vu}_{j1} \right\|_2}} = 1.
\end{equation}
With analogous derivation, we can prove
\begin{equation}\label{Eq:uXih}
\mathop {\lim }\limits_{M,C \to \infty } \frac{e^{-j\phi}\widetilde{\vu}_{j1}^H {{\bf{\Xi}}_j'}^H {{\bf{\Xi}}_j'} \underline{\vh}_j}{{\left\| {\bf{\Xi}}_j' \underline{\vh}_j \right\|_2 {\left\| {\bf{\Xi}}_j' \widetilde{\vu}_{j1} \right\|_2}}} = 1.
\end{equation}
Applying (\ref{Eq:hXiu}) and (\ref{Eq:uXih}) to (\ref{Eq:h_udl_converge2}) gives
\begin{equation}
\mathop {\lim }\limits_{M,C \to \infty } {\left\| \frac{{\bf{\Xi}}_j' \underline{\vh}_j}{\left\| {\bf{\Xi}}_j' \underline{\vh}_j \right\|_2} - \frac{{\bf{\Xi}}_j' \widetilde{\vu}_{j1} e^{j\phi}}{\left\| {\bf{\Xi}}_j' \widetilde{\vu}_{j1} \right\|_2}  \right\|_2^2 } = 0.
\end{equation}
The following equality holds
\begin{equation}
{\bf{\Xi}}_j' \underline{\vh}_j = {\bf{\Xi}}_j' {\bf{\Xi}}_j {\vh}_j = \mR_j^\dag \mR_j {\vh}_j = {\vh}_j,
\end{equation}
proving that
\begin{equation}
\mathop {\lim }\limits_{M,C \to \infty } {\left\| \frac{{\vh}_j}{\left\| {\vh}_j \right\|_2} - \bar{\bf{u}}_{j1} e^{j\phi} \right\|_2 } = 0,
\end{equation}
which completes the proof of Lemma \ref{lemmaEigenvectorOrig}.
\end{proof}

\subsection{Proof of Theorem \ref{theoStatis}:}\label{proof:theoStatis}
\begin{proof}
From (\ref{Eq:h_converge}) we readily obtain
\begin{equation}\label{Eq:hu_converge}
\mathop {\lim }\limits_{M,C \to \infty }  \frac{{\vh}_j^H \bar{\bf{u}}_{j1}}{\left\| {\vh}_j \right\|_2}  = 1.
\end{equation}
Recall from the uplink training (\ref{Eq:train}), we have
\begin{equation}\label{Eq:CAh1}
{\bf{\widehat h}}_{j}^{\text{CA}} = \frac{1}{\tau} \bar{\bf{u}}_{j1} \bar{\bf{u}}_{j1}^{H} \left( {\vh}_j \vs^T + \sum_{l\neq j}{\vh}_l \vs^T + {\bf{N}} \right)  \vs^*,
\end{equation}
and hence
\begin{align}
& \quad \mathop {\lim }\limits_{M, C \to \infty } \frac{\left\|{\bf{\widehat h}}_{j}^{\text{CA}} - \vh_{j} \right\|^2_2}{\left\|\vh_{j}\right\|^2_2}  
\end{align}
\begin{align}
& = \mathop {\lim }\limits_{M, C \to \infty } \frac{ ({\bf{\widehat h}}_{j}^{\text{CA}} - \vh_{j} )^H ({\bf{\widehat h}}_{j}^{\text{CA}} - \vh_{j} )}{\left\|\vh_{j}\right\|^2_2} \nonumber \\
& = \mathop {\lim }\limits_{M, C \to \infty } \frac{1}{\left\|\vh_{j}\right\|^2_2} \left( \sum_{l\neq j}{\vh}_l\bar{\bf{u}}_{j1} \bar{\bf{u}}_{j1}^H \sum_{l\neq j}{\vh}_l + \sum_{l\neq j}{\vh}_l\bar{\bf{u}}_{j1} \bar{\bf{u}}_{j1}^H {\bf{N}} \frac{\vs^*}{\tau} \right. \nonumber \\
& \qquad + \left. \frac{\vs^T}{\tau}{\bf{N}}^H \bar{\bf{u}}_{j1} \bar{\bf{u}}_{j1}^H \sum_{l\neq j}{\vh}_l + \frac{\vs^T}{\tau}{\bf{N}}^H \bar{\bf{u}}_{j1} \bar{\bf{u}}_{j1}^H {\bf{N}} \frac{\vs^*}{\tau} \right. \nonumber \\
& \qquad - \vh_{j}^H \bar{\bf{u}}_{j1} \bar{\bf{u}}_{j1}^H \vh_{j} + \vh_{j}^H \vh_{j} \Big)\nonumber \\
& = \mathop {\lim }\limits_{M, C \to \infty } \frac{1}{\left\|\vh_{j}\right\|^2_2} \left( \vh_{j}^H \vh_{j} - \vh_{j}^H \bar{\bf{u}}_{j1} \bar{\bf{u}}_{j1}^H \vh_{j} \right).
\end{align}
Equation (\ref{Eq:hu_converge}) ensures that
\begin{align}
\mathop {\lim }\limits_{M, C \to \infty } \frac{1}{\left\|\vh_{j}\right\|^2_2} \vh_{j}^H \bar{\bf{u}}_{j1} \bar{\bf{u}}_{j1}^H \vh_{j} =  \frac{1}{\left\|\vh_{j}\right\|^2_2} \vh_{j}^H \vh_{j} = 1,
\end{align}
which concludes the proof.
\end{proof}

\subsection{Proof of Theorem \ref{theoAnglePower}:}\label{proof:theoAnglePower}
\begin{proof}
This proof follows similar steps towards Theorem \ref{theoStatis}. Thus we give a sketch of the proof only. Define
\begin{align}\label{Eq:Gamma2}
\Gamma &\triangleq \mathop {\lim }\limits_{C \to \infty } \left({\frac{1}{C} { {\bf{\widetilde W}}_j {\bf{\widetilde W}}_j^{H}} } \right) \\
& = {{\underline{\mathbf{h}}}_{j}^{(j)}} {{\underline{\mathbf{h}}}_{j}^{(j)H}} + \sum_{l \neq j} {{\underline{\mathbf{h}}}_{l}^{(j)}} {{\underline{\mathbf{h}}}_{l}^{(j)H}} + \sigma_n^2 {\bf{\Xi}}_j{\bf{\Xi}}_j^H,
\end{align}
where ${{\underline{\mathbf{h}}}_{l}^{(j)}} \triangleq {\bf{\Xi}}_j {{{\mathbf{h}}}_{l}^{(j)}}, l = 1, \ldots, L$.
Due to the asymptotic orthogonality between steering vectors in disjoint angular support, i.e., Lemma 3 in \cite{yin:13}, we can easily show that in large antenna limit, ${{\mathbf{h}}_{l\text{o}}^{(j)}}$ falls into the null space of ${{\bf{R}}_{j}^{(j)}}$. Thus
\begin{align}
\mathop {\lim }\limits_{M \to \infty } \frac{1}{M} {{\underline{\mathbf{h}}}_{l}^{(j)}} {{\underline{\mathbf{h}}}_{l}^{(j)H}} = \mathop {\lim }\limits_{M \to \infty } \frac{1}{M} {{\underline{\mathbf{h}}}_{li}^{(j)}} {{\underline{\mathbf{h}}}_{li}^{(j)H}}.
\end{align}
Then we have
\begin{align*}
\mathop {\lim }\limits_{M \to \infty } \frac{\Gamma}{M} = \frac{1}{M}\left( {{\underline{\mathbf{h}}}_{j}^{(j)}} {{\underline{\mathbf{h}}}_{j}^{(j)H}} + \sum_{l \neq j} {{\underline{\mathbf{h}}}_{li}^{(j)}} {{\underline{\mathbf{h}}}_{li}^{(j)H}} + \sigma_n^2 {\bf{\Xi}}_j{\bf{\Xi}}_j^H \right).
\end{align*}
Under Condition C1, it is easy to show that
\begin{align}
\mathop {\lim }\limits_{M \to \infty } \left\| \frac{\Gamma}{M} \frac{{{\underline{\mathbf{h}}}_{j}^{(j)}}}{\left\|{{\underline{\mathbf{h}}}_{j}^{(j)}}\right\|_2} - \frac{{{\underline{\mathbf{h}}}_{j}^{(j)^H}} {{\underline{\mathbf{h}}}_{j}^{(j)}}}{M} \frac{{{\underline{\mathbf{h}}}_{j}^{(j)}}}{\left\|{{\underline{\mathbf{h}}}_{j}^{(j)}}\right\|_2} \right\|_2 &= 0.
\end{align}
Given the following condition
\begin{equation}
\forall l \neq j, \left\| {\bf{\Xi}}_j {{\mathbf{h}}_{l\text{i}}^{(j)}} \right\|_2 < \left\|{\bf{\Xi}}_j {{\mathbf{h}}_{j}^{(j)}} \right\|_2,
\end{equation}
it is clear that the dominant eigenvector of ${\Gamma}/{M}$ converges to ${{{\underline{\mathbf{h}}}_{j}^{(j)}}}/{\left\|{{\underline{\mathbf{h}}}_{j}^{(j)}}\right\|_2}$ (up to a random phase), with its corresponding eigenvalue converging to ${{{\underline{\mathbf{h}}}_{j}^{(j)^H}} {{\underline{\mathbf{h}}}_{j}^{(j)}}}/{M} $. Then, using the same technique in the proof of Lemma \ref{lemmaEigenvectorOrig}, we obtain
\begin{equation}\label{Eq:hu_converge2}
\mathop {\lim }\limits_{M,C \to \infty }  \frac{{{\mathbf{h}}_{j}^{(j)^H}} \bar{\bf{u}}_{j1}}{\left\| {{\mathbf{h}}_{j}^{(j)}} \right\|_2}  = 1.
\end{equation}
Finally, we readily obtain (\ref{Eq:limit_partialOverlap}) by analogous derivations in Appendix \ref{proof:theoStatis}.
\end{proof}

\bibliography{bib/allCitations}
\bibliographystyle{IEEEtran}
\begin{IEEEbiography}[{\includegraphics[width=1in,height=1.25in,clip,keepaspectratio]{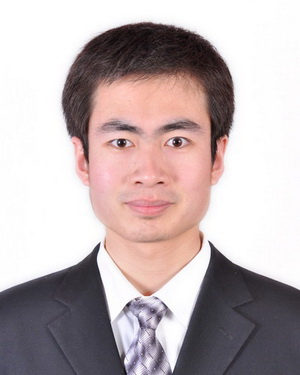}}]
{Haifan Yin}
received the B.Sc. degree in Electrical and Electronic Engineering and the M.Sc. degree in Electronics and Information Engineering from Huazhong University of Science and Technology, Wuhan, China, in 2009 and 2012 respectively. From 2009 to 2011, he had been with Wuhan National Laboratory for Optoelectronics, China, working on the implementation of TD-LTE systems as an R\&D engineer. In September 2012, he joined the Mobile Communications Department at EURECOM, France, as a Ph.D. student. His current research interests include signal processing, channel estimation, cooperative networks, and large-scale antenna systems.

H. Yin was a recipient of the 2015 Chinese Government Award for Outstanding Self-financed Students Abroad.
\end{IEEEbiography}
\vspace{-0.5cm}

\begin{IEEEbiography}[{\includegraphics[width=1in,height=1.25in,clip,keepaspectratio]{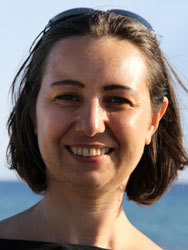}}]
{Laura Cottatellucci} is currently assistant professor at the Dept. of Mobile Communications in Eurecom. She obtained the PhD from Technical University of Vienna, Austria (2006). Specialized in networking at Guglielmo Reiss Romoli School (1996, Italy), she worked in Telecom Italia (1995-2000) as responsible of industrial projects. From April 2000 to September 2005 she worked as senior research in ftw Austria on CDMA and MIMO systems. From October to December 2005 she was research fellow on ad-hoc networks in INRIA (Sophia Antipolis, France) and guest researcher in Eurecom. In 2006 she was appointed research fellow at the University of South Australia, Australia to work on information theory for networks with uncertain topology. Cottatellucci is co-editor of a special issue on cooperative communications for EURASIP Journal on Wireless Communications and Networking and co-chair of RAWNET/WCN3 2009 in WiOpt.  Her research topics of interest are large system analysis of wireless and complex networks based on random matrix theory and game theory.
\end{IEEEbiography}

\vspace{-0.5cm}
\begin{IEEEbiography}[{\includegraphics[width=1in,height=1.25in,clip,keepaspectratio]{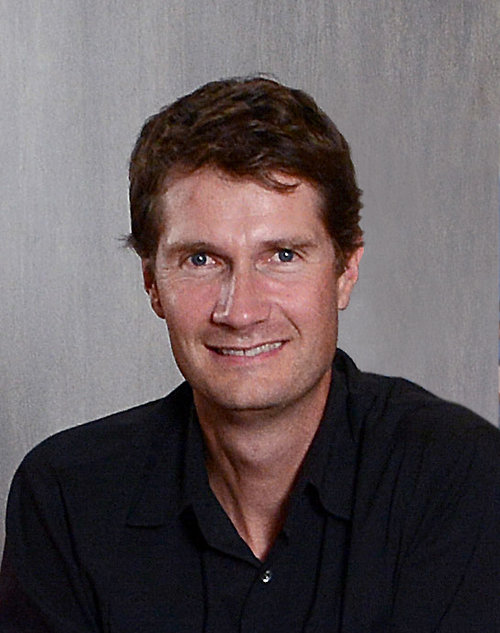}}]
{David Gesbert} (F'11) is Professor and Head of the Mobile Communications Department, EURECOM, France. He obtained the Ph.D degree from Ecole Nationale Superieure des Telecommunications, France, in 1997.
From 1997 to 1999 he has been with the Information Systems Laboratory, Stanford University. In 1999, he was a founding engineer of Iospan Wireless Inc, San Jose, Ca.,a startup company pioneering MIMO-OFDM (now Intel). Between 2001 and 2003 he has been with the Department of Informatics, University of Oslo as an adjunct professor. D. Gesbert has published about 250 papers and several patents all in the area of signal processing, communications, and wireless networks. He was named in the 2014 Thomson-Reuters List of Highly Cited Researchers in Computer Science. Since 2015, he holds an ERC advanced grant on the topic of ``Smart Device Communications"

  D. Gesbert was a co-editor of six special issues on wireless networks and communications theory.  He was an associate editor for IEEE Transactions on Wireless Communications and the EURASIP Journal on Wireless Communications and Networking. He authored or co-authored papers winning the 2004 and 2015 IEEE Best Tutorial Paper Award (Communications Society), 2012 SPS Signal Processing Magazine Best Paper Award. He co-authored the book ``Space time wireless communications: From parameter estimation to MIMO systems", Cambridge Press, 2006.  He is a Technical Program Chair for IEEE ICC 2017, to be held in Paris.

\end{IEEEbiography}


\begin{IEEEbiography}[{\includegraphics[width=1in,height=1.25in,clip,keepaspectratio]{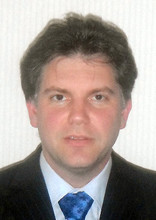}}]{Ralf R. M\"uller}
(S'96--M'03--SM'05) was born
in Schwa\-bach, Germany, 1970. He received the
Dipl.-Ing.\ and Dr.-Ing.\ degree with distinction from Friedrich-Alexander-Universit\"at (FAU)
Erlangen-N\"urnberg in 1996 and 1999, respectively.
From 2000 to 2004, he directed a research group
at The Telecommunications Research Center Vienna in Austria and taught as an adjunct professor
at TU Wien. In 2005, he
was appointed full professor at the Department of
Electronics and Telecommunications at the Norwegian University of Science
and Technology in Trondheim, Norway. In 2013, he joined the Institute for
Digital Communications at FAU Erlangen-N\"urnberg in Erlangen, Germany.
He held visiting appointments at Princeton University, US, Institute Eur\'ecom,
France, University of Melbourne, Australia, University of Oulu, Finland,
National University of Singapore, Babe{\c s}-Bolyai University, Cluj-Napoca,
Romania, Kyoto University, Japan, FAU Erlangen-N\"urnberg, Germany, and
TU M\"unchen, Germany.

Dr.\ M\"uller received the Leonard G.\ Abraham Prize (jointly with Sergio
Verd\'u) for the paper ``Design and analysis of low-complexity interference
mitigation on vector channels" from the IEEE Communications Society. He
was presented awards for his dissertation ``Power and bandwidth efficiency of
multiuser systems with random spreading" by the Vodafone Foundation for
Mobile Communications and the German Information Technology Society
(ITG). Moreover, he received the ITG award for the paper ``A random matrix
model for communication via antenna arrays" as well as the Philipp-Reis
Award (jointly with Robert Fischer). Dr. M\"uller served as an associate editor
for the IEEE TRANSACTIONS ON INFORMATION THEORY from 2003 to 2006. He is currently serving on the executive editorial board of the IEEE TRANSACTIONS ON WIRELESS COMMUNICATIONS.
\end{IEEEbiography}

\vspace{-10cm}

\begin{IEEEbiography}[{\includegraphics[width=1in,height=1.25in,clip,keepaspectratio]{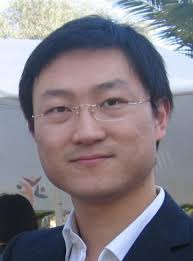}}]{Gaoning He} received his Master and Ph.D degrees in France, from \'Ecole d'Ing\'enieurs T\'el\'ecom ParisTech (ENST). He had been a research engineer in Motorola research center in Paris from 2006 to 2009, research fellow in Alcatel-Lucent Chair on Flexible Radio - Sup\'elec in Paris from 2009 to 2010, and research scientist at Alcatel-Lucent Bell-labs in Shanghai from 2010 to 2011. He joined Huawei Shanghai research center as senior researcher in 2011. Since 2014, he is assistant director and research manager of the Mathematical and Algorithmic Sciences Lab in Huawei France Research Center. His research interests lie in fundamental mathematics, algorithms, statistics, information \& communication sciences research.
\end{IEEEbiography}


\end{document}